\documentclass[a4paper,10pt,oneside]{book}
\usepackage[a4paper, left=2cm, right=2cm, top=2.5cm, bottom=2.5cm, headsep=1.2cm]{geometry} 
\usepackage[utf8x]{inputenc}
\usepackage[english,polish]{babel}
\usepackage{polski}
\usepackage[pdftex]{graphicx}
\usepackage{amsmath}
\usepackage{amsthm}
\usepackage{amsfonts}
\usepackage{latexsym}
\usepackage{indentfirst}
\usepackage{floatrow}
\usepackage{enumerate}
\usepackage{bm}
\usepackage{titlesec}

\makeatletter
\newcommand{\linia}{\rule{\linewidth}{0.4mm}}

\renewcommand{\maketitle}{ 

    \begin{titlepage}

    \begin{center}

    {\Large Institute of Mathematics \\
		 Faculty of Mathematics and Computer Science\\
		Jagiellonian University\\}

    \end{center}

    \vspace{3.5cm}

		\begin{center}

    \LARGE Master Thesis

    \end{center}
    
    \noindent\linia

    \begin{center}

      \huge \textsc{\@title}

     \end{center}

    \noindent \linia

    \vspace{0.5cm}
		
		\begin{center}
		
		\LARGE \@author \par

		\end{center}

    \vspace{0.1cm}

		\begin{center}
		
    {\Large Supervisor: Dominik Kwietniak}\\
    
		\end{center}
		
    \vspace*{\stretch{6}}
		
		\vspace{2cm}
    
    \begin{center}
    
    \begin{figure}[!h]
    \includegraphics[height=4cm]{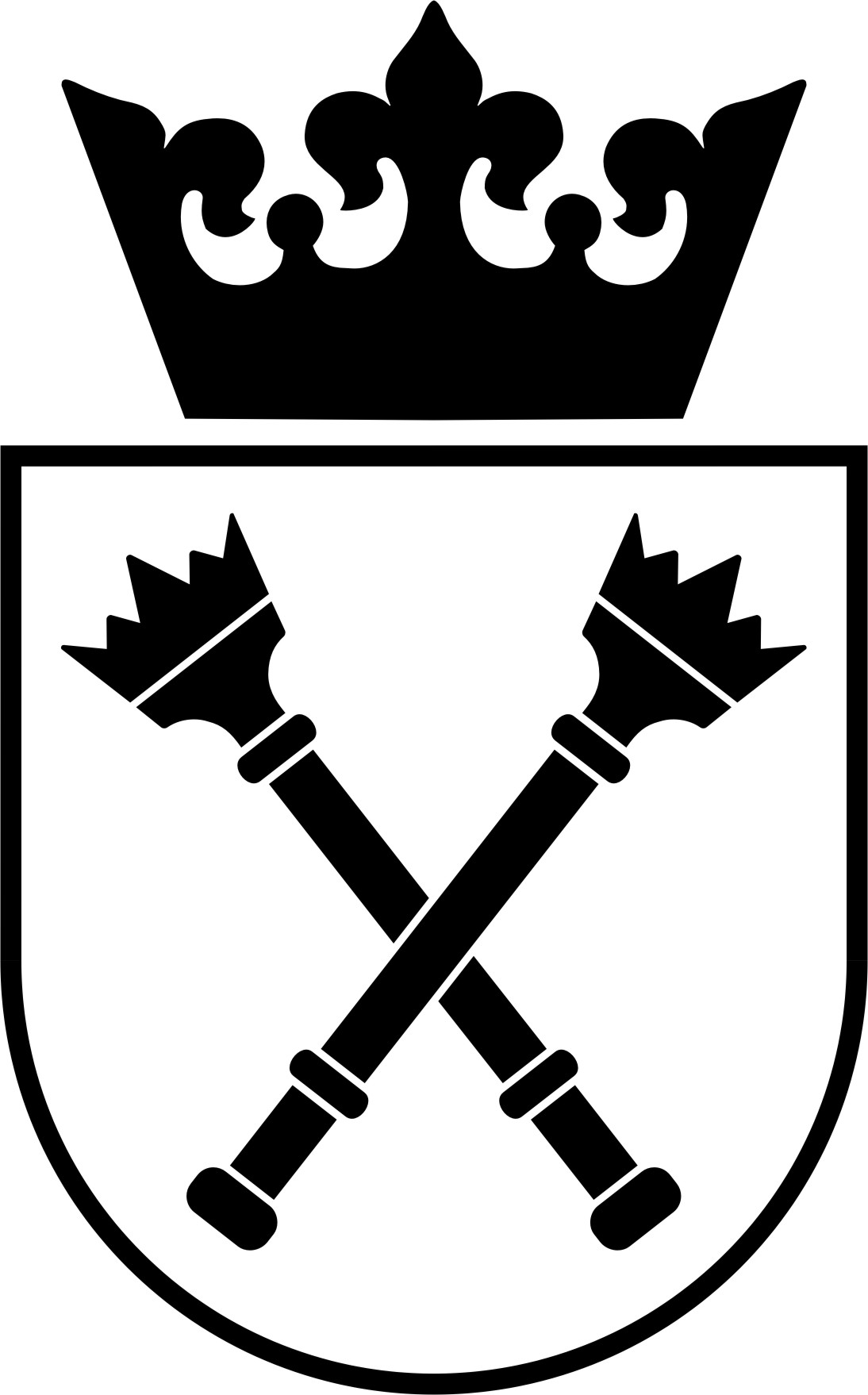}
    \end{figure}

    \Large Cracow, 2013

    \end{center}

  \end{titlepage}%

}

\makeatother

\title{Mathematical models for epidemic spreading \\ on complex networks}
\author{Wojciech Ganczarek}

\begin{document}
\selectlanguage{english}
\renewcommand*{\figurename}{FIG.}
\newtheorem{thm}{Theorem}
\newtheorem{lemma}{Lemma}
\newtheorem{df}{Definition}
\newtheorem{prop}{Property}
\newtheorem{remark}{Remark}
\newtheorem{ex}{Example}
\newfloatcommand{capbtabbox}{table}[][\FBwidth]

\maketitle

\chapter*{\centering Abstract}
\addcontentsline{toc}{chapter}{Abstract}
We propose a model for epidemic spreading on a finite complex network with a restriction to at most one contamination per time step. Because of a highly discrete character of the process, the analysis cannot use the continous approximation, widely exploited for most of the models. Using discrete approach we investigate the epidemic threshold and the quasi-stationary distribution. The main result is a theorem about mixing time for the process, which scales like logarithm of the network size and which is proportional to the inverse of the distance from the epidemic threshold. In order to present the model in the full context, we review modern approach to epidemic spreading modeling based on complex networks and present necessary information about random networks, discrete-time Markov chains and their quasi-stationary distributions. 

\vspace{3.5cm}
\noindent \linia 
\vspace{0.5cm}

\selectlanguage{polish}
\begingroup
\let\clearpage\relax
\chapter*{\centering Streszczenie}
W niniejszej pracy proponuję model rozprzestrzeniania się epidemii na skończonej sieci złożonej z ograniczeniem do co najwyżej jednego zarażenia na krok czasowy. Ze względu na typowo dyskretny charakter procesu, w~analizie nie można zastosować przybliżenia ciągłego, które używa się w odniesieniu do większości modeli. Stosując podejście dyskretne, badam próg epidemii i rozkład kwazi-stacjonarny. Głównym rezultatem pracy jest twierdzenie dotyczące czasu mieszania dla omawianego procesu, który rośnie logarytmicznie z rozmiarem sieci i jest odwrotnie proporcjonalny do odległości od progu epidemii. W celu przestawienia proponowanego modelu w pełnym kontekście, omawiam współczesne podejście do modelowania epidemii na sieci złożonej i prezentuję konieczne informacje dotyczące sieci losowych, łańcuchów Markowa o czasie dyskretnym oraz ich rozkładów kwazi-stacjonarnych. 

\endgroup
\selectlanguage{english}

  \tableofcontents

\chapter*{Preface}
\addcontentsline{toc}{chapter}{Preface}
The problem of epidemic spreading has always been an important issue in the society. Especially in the past, the human race suffered a lot because of highly contagious diseases. One of the most significant examples is the Black Death, the pandemic of bubonic plague that killed 1/4 of the 100 million population of Europe in 14th century. Mathematical approach to this problem has also quite long history, most probably initiated by D.Bernoulli in 1760 \cite{may}. However, we can also think of epidemic spreading among plants, animals or computers. In the 20th century the problem was studied intensively: scientists reduced this highly complicated issue to sets of differential equations and a few parameters like rate of contact between susceptible and infectious individuals, and variables like density of susceptible people etc. A broad review of the literature associated with classical approach to mathematical modeling can be found e.g. in \cite{may,bailey}. 

The basic questions of mathematical modeling of epidemic spreading are closely related to the ones we ask while solving problems of the public health. That is: in what kind of conditions we can expect epidemic outbreak? How many people will be infected afterwards? How much time does it take to stabilize this process? All of those questions are extremely hard to answer in reality. There are so many factors that affect epidemic spreading: the way of contamination of a particular virus/bacteria, the frequency and character of contacts between people, duration of the disease, resistance or immunity of particular individuals or, for instance, the weather. This is why solving epidemiological problems involves researchers from many disciplines, such as biology, computer science, social sciences, physics and, last but not least, mathematics.

Nowadays, along with the advent of the complex system science (see \cite{NewmanSIAM} for a review), there has been a vivid development of the mathematical modeling for epidemic spreading based on complex networks. The most apparent, and at the same time the most significant change, is that in classical approach mentioned above each individual could contact and contaminate any other -- which is not the most realistic assumption. Now, when network structure comes into play, the contacts are restricted to neighboring nodes only, which for instance in the context of social networks can be thought of as acquaintances. The structure of connections between individuals appears to affect dynamics of the process significantly, what we will see especially in the case of scale-free networks.

Networks seem to be an omnipotent framework for solving different many-body problems from almost all the branches of science: epidemic, gossip or in general: information spreading, technological networks problems (for instance: telephone, internet or power grid), neural, ecological or biochemical networks, and much more, mainly related to society, biology and computer science. 
Giving promising results, the network-related approach appear to be the right way of investigating complex systems and will surely play an important role in the contemporary science. Especially because the need of solutions for complex system rises, and as Stephen Hawking said: {\sl ``I think the next century will be the century of complexity''}. Also the behavior of a network itself is an interesting field of research: formation of networks, network search or privicy has been widely studied in last 10--20 years (see \cite{BarabasiRMP,Networks} for a review). In a sense, the investigation towards epidemic spreading goes now on in two ways: one of them is to understand the structure of social networks, and the second: development of mathematical models that drives dynamics on them.

The aim of this dissertation is to provide an introduction of the modern approach to epidemic spreading in the context of complex networks and stochastic processes. Moreover, we propose a new way of tackling epidemiological problems by a single infection epidemic spreading model. Obviously, this elaboration is far from being complete, but an interested reader can acces the rapidly developing realm of complex system by the wide range of references that we provide. The text is organized as follows. In Sec.~\ref{introduction} we give a review of models for random networks, which are used to mimic the real ones in epidemic spreading and other kinds of modeling. In Sec.~\ref{markov} some rudiments of stochastic processes and Markov chains are presented. Moreover, we also mention some facts about quasi-stationary distributions, which are useful to describe long time behavior of epidemies. Sec.~\ref{models} is dedicated to modern models for epidemic spreading, while Sec.~\ref{abm} is focused on the model that we propose.
\chapter{Network models}
\label{introduction}
Scientists representing a broad range of fields -- from biology and chemistry, telecommunication and electric engineering, up to sociology, economy and political sciences -- deal with higlhy complicated, multicomponent, interacting systems. It could be a food web, a metabolic system, a telephone network or power grid, the problem of gossip spreading, crisis on financial market or changes of people's decisions during parliamentary elections. All of the examples above, although associeted with totally different phenomena, share at least one common feature: they are unbelievebly complex. Think of an example of parliamentary elections. What does the decision of a particular voter depend on? Party, beliefs, religion, financial situation, appearance of candidates... Then add interactions: people talk with each other, they try to convince their neighbors. Moreover: all of them are influenced by media. And then finally, in the day of elections, the wheather breaks down and our voter stays home and does not vote at 
all. How anyone can predict the effect of common behavior of millions of such voters? That really sounds impossible. The point is to divide the problem into a group of managable tasks. One for sure would be, sticking to our elections example, to investigate people's opinions. Then, to classify their {\sl interactions}, i.e. to distinguish different ways of passing information. Then, and this is the topic of the present chapter, we have to know the {\sl pattern} of these interactions. We know what kind of information can be spread out (''Vote for this and that guy!''), we know how this information can be spread out (e.g. during a long evening in a local pub) but, what is probably the most important thing if we would like to treat the problem globally, we do not know the routes that a piece information can pass. The search of pattern of interaction is basically a reduction of a whole system to a simple network: to points connected by some links. In the example of voting it could be a social network, i.e. a 
network 
of people, who are connected between each other by relationships, common workplaces or, in general, any feature that makes it possible to transfer information between them. It turns out, that the structure of the network -- the {\sl pattern of interactions} -- plays often a crucial role in dynamics of processes going on in many different system, e.g. those listed above: food webs, power grids, financial markets etc. For instance, social networks, because of their structure, appear to be very weak in the face of epidemy, what we describe in details in Chapter \ref{models}. This is why we need to understand the structures and patterns that govern the networks. 

\begin{figure}[!ht]
  \includegraphics[height=8cm]{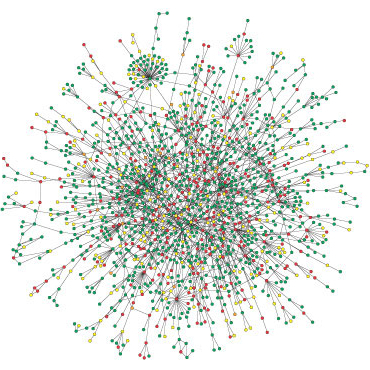}
  \label{fig:graph}
  \caption{An example of network representation of a complex system: a protein-protein interaction network for baker's yeast ({\sl Saccharomyces cerevisiae}). Picture taken from \cite{protein}.}
\end{figure}

Collecting detailed data about all networks we would need, is not always an easy task and quite often it is somehow expensive. Using available datasets people try to construct models of networks that mimic real systems in order to understand the structure and mechanisms that govern evolution of these structures. Real and model networks are used to predict different phenomena, for instance: epidemic spreading on a social network. In the present chapter we describe three of such network models: Erd\"os-Rényi random graph, Small Worlds and Scale-free networks. 

An outstanding, broad review of complex system phenomena can be found in \cite{ball}. In particular, a good prediction of vote distribution in Brazilian elections is given in \cite{vote}.

\section{Graph theory preliminaries}
\label{graph}
Before we describe network models, we need some definitions from the graph theory. We represent given network by a graph, which, in the simplest case, is just a set of nodes and connections between them. The structure of graphs can be, however, much richer -- there are also weighted or directed graphs, which we introduce in futher sections. As graph theory serves us as a framework for processes on complex networks, we try to present it in this context by giving some real-world examples. We start with a definition of a graph, its subgraph and the most basic relation between two nodes: neighborhood. In this elaboration we restrict our considerations to finite graphs. The terminology follows \cite{gt,gt2}.

\begin{df} A graph $G$ is a pair $G=(V,E)$, where $V$ is a non-empty finite set and $E$ is a subset of the set of unordered pairs of elements of $V$. The elements of E are called {\sl edges} (or: links, bonds), while the elements of $V$ are called {\sl vertices} (or: nodes, sites). A graph $G'=(V',E')$ is a subgraph of the graph $G=(V,E)$ if $V'\subseteq V$ and $E' \subseteq E$. The order of a graph $G=(V,E)$, denoted by $|G|$, is the cardinality of the set V.
\end{df}

\begin{df} If $\{i,j\}\in E$ we say that $i,j$ are neighbors (or: that they are adjacent). The neighborhood of a vertex $v$ is the set of vertices adjacent to $v$, we denote it by $N(v)$. The degree of a vertex, denoted by $deg(v)$, is the number of its neighbors.  %Analogous definitions hold for directed graphs.
\end{df}

\begin{remark}
In a graph $G$, the sum of the degrees of the vertices is equal to twice the number of edges. 
The maximum number of edges for a graph $G=(V,E)$ is $\binom{N}{2}$, where $N=|G|$.
\end{remark}

In various realizations of networks we observe links, that point from one node to another, but do not point back. A real-world example is a food web ({\sl ''what eats what''}), where in particular storks eat snails, but not {\sl vice versa}.

\begin{df} A directed graph $G$ is a pair $G=(V,\vec{E})$, where $V$ is a non-empty finite set and $\vec{E}$ is a subset of the set of ordered pairs of elements of $V$. 
\end{df}

In the following we adopt a convention, that by ``graph'' we mean this kind of graph (undirected, directed or both) that makes sense in a given context. The same applies to $\vec{E}$, which for undirected graph mean the same as $E$. A useful representation of graphs, particularly convinient for algebraic manipulations, is given by the notion of {\sl adjacency matrix}. It is however worth stressing, that in large scale computations it is much more efficient to work on edge lists, i.e. sequences of pairs of vertices that are connected. 

\begin{df}
For a given graph $G=(V,\vec{E})$ of the order N the adjacencey matrix $\mathbb{A}=[a_{ij}]_{N\times N}$ is an $N\times N$ matrix defined by:
$$a_{ij}=\left\{ \begin{array}{ll}
                  1, & \textrm{if } (i,j)\in \vec{E}, \\
		   0, & \textrm{if } (i,j)\notin \vec{E}. 
                 \end{array} \right. 
$$ 
\end{df}

\begin{remark}
Adjacency matrix of an undirected graph is symmetric.  
\end{remark}

Consider now the following example from social networks: a 5 years old son and his father know each other, but opinions of father influence his son much more than he is influenced by his son's opinions (at least: usually). In this case a graph made of simple nodes and links only would not be an accurate picture of reality. We need here to put some {\sl weights} to represent the power of influence between people.

\begin{df}
Given a graph $G=(V,\vec{E})$, a {\sl weight function} is a function $\phi$ that maps the edges of $G$ to the non-negative real numbers, i.e.:
$$\phi:\vec{E}\to \mathbb{R}_+.$$
A pair $(G,\phi)$ is called a {\sl weighted graph}.
\end{df}

The crucial notion for epidemic spreading on graphs is the {\sl connectivity}. Below we introduce the notions of a path, connectivity between nodes and finally: connectivity of a graph.

\begin{df}
A path $P_{i_0,i_n}$ of the length $n$ from $i_0$ to $i_n$ in a graph $G=(V,E)$, is a sequence of $n+1$ vertices $(i_0,i_1,...,i_n)$ such that $\{i_k,i_{k+1}\}\in E$ for all $k=0,...,n-1$. A directed path $\vec{P}_{i_0,i_n}$ of the length $n$ from $i_0$ to $i_n$ in a directed graph $G=(V,\vec{E})$, is a sequence of $n+1$ vertices $(i_0,i_1,...,i_n)$ such that $(i_k,i_{k+1})\in \vec{E}$ for all $k=0,...,n-1$. If for a (directed) path we have $i_0=i_n$, when this path is called a (directed) loop.
\end{df}

\begin{df}
Nodes $i,j$ of a graph $G=(V,E)$ are said to be {\sl connected} if there exists a path from $i$ to $j$, we write then $i\leftrightarrow j$. Consider nodes $i,j$ of a graph $G=(V,\vec{E})$. Node $j$ is said to be reachable from $i$ if there exists a directed path $\vec{P}_{i,j}$ between them, we write then $i\rightarrow j$. 
\end{df}

\begin{df}
We say that $G=(V,E)$ is a connected graph if there exists a path between each pair of nodes of this graph, i.e. for all $i,j$ there is $i\leftrightarrow j$.  \\
\indent
We say that a directed graph $G=(V,\vec{E})$ is strongly connected if there exists a directed path between each ordered pair of nodes, i.e. for all $i,j$ there is $i\rightarrow j$. Furthermore, we say that a directed graph is weakly connected if by replacing all of its directed edges by undirected edges we get a connected undirected graph.
\end{df}

Given a graph it is {\sl a priori} not embedded in a metric space. Considering some real-world networks we can find a natural embedding for them, e.g. in case of a graph of airports or metro stations. In some cases it is also possible to search for so-called hidden metric spaces (see e.g. \cite{hid1,hid2}). But the most basic quantities that measure distances on a graph and its size are the {\sl shortest path} and {\sl diameter}.

\begin{df}
A geodesic path is the shortest path between two vertices.
\end{df}

\begin{df}
The diameter of a graph is the length of the longest geodesic path between any pair of connected vertices in the network.
\end{df}

There is yet another important quantity that characterizes a graph, called conductance \cite{Con}. It shows how well-connected a given graph is. Note first, that any graph $G$ defined by the adjacency matrix $\mathbb{A}$ can be equivalently described by a stochastic matrix $P=[p_{ij}]_{N\times N}$ defined as $p_{ij}=\frac{a_{ij}}{\sum_{j=1}^N a_{ij}}$ if $\sum_{j=1}^N a_{ij}>0$, or $0$ otherwise.

\begin{df} 
Conductance of a given graph $G$ described by a stochastic matrix $\{p_{ij}\}$ is:
\begin{equation}
\Phi (P)=\min_{S\subset V}\frac{\sum\limits_{j\in S,\,k\notin S}p_{jk}}{min\{|S|,|V-S|\}}.  
\end{equation}
\label{df:con}
\end{df}

After this handful of definitions, we provide two examples of graph.

\begin{ex}
Let us consider an example of a graph $G=(V,E)$ of the order 4, depicted in Fig.\ref{fig:graph}.
This graph is directed, strongly connected and unweighted. Using labels from the picture we 
can write \linebreak $E=\{(1,2),(2,1),(2,4),(3,2),(3,4),(4,3)\}$. The adjacency matrix of G is:
$$\mathbb{A}=\left(\begin{tabular}{cccc}
 0 & 1 & 0 & 0 \\
 1 & 0 & 0 & 1 \\
 0 & 1 & 0 & 1 \\
 0 & 0 & 1 & 0
\end{tabular}\right) .$$

\begin{figure}
  \includegraphics[height=4cm]{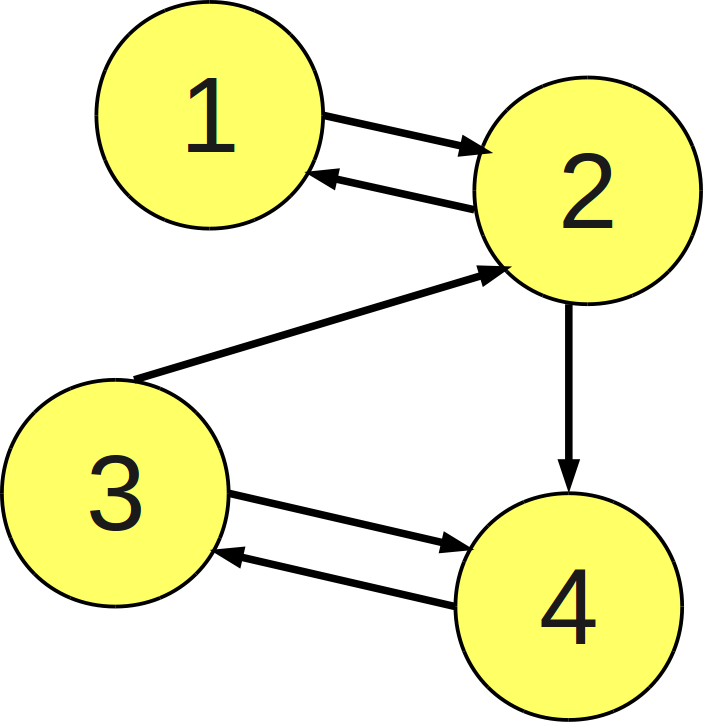}
  \caption{An example of directed and strongly connected graph.}
    \label{fig:graph}
\end{figure}

\end{ex}

Another example is a complete graph.

\begin{df}
We say that a graph $G$ is a complete (or full) graph if each node is a neighbor of all other nodes.
\end{df}

A typical feature of real network is clustering, for instance: people tend to form groups of friends. In such a group they know each other well -- we would have thus a lot of nodes between vertices of the subgraph representing this group -- but they are usually poorly connected with the people outside the group which we consider. We would like to develop a quantity to measure such a feature. For this aim we employ a factor called \textit{clustering coefficient} $C$ \cite{Networks}. By $C$ we can quantify \textit{transitivity} of a network, namely the feature that if vertices $i$, $j$ are adjacent, and $j$, $k$ are adjacent too, then it follows that $i,k$ are adjacent as well. In the sociological context we can say: {\textit ``The friend of my friend is also my friend''}. Let us translate it into graph-related notions: consider a set of nodes $\{i,\,j,\,k\}$. If $i,j$ and $j,k$ are mutual neighbors, we have a path of the lenght 2. If, additionally, $i,k$ are neighbors, then we have a \textit{loop} of the 
length 3 (triangle). Thus the idea will be to compare the number of loops of the lenght three with all the paths of length two.
\begin{df}
We define clustering coefficient C as:
\begin{equation}
   C=\frac{\#\{\textrm{triangles}\}\times 6}{\#\{ \textrm{directed paths of length 2}\}}.
\end{equation}
   \label{df:cc2}
\end{df}
The factor $6$ arises because a triangle $ijk$ contains $6$ paths of length 2: $ijk$, $jki$, $kij$ plus 3 paths going the opposite direction. There are also different ways of defining $C$ (not always equivalent), see e.g. in \cite{Networks}.

Adjacency matrix gives us full information about a graph. If we however would like to merge graphs in some classes, we need more general notions to describe them. For instance, one can imagine an ensamble of graphs with different adjacency matrices, but all of them having node degrees given by the same probability distribution. 

\begin{df}
The degree distribution is a probability distribution of degrees given by $\mathbb{P}(k)$ the fraction of nodes of a network with degree $k$, $\mathbb{P}(k)=\frac{N_k}{N}$.
\end{df}
That is, we consider the degree of a random vertex as a random variable. Therefore, in particular $\mathbb{E}(k)=\sum_k k\mathbb{P}(k)$ denotes the expected degree of a randomly chosen graph. We usually expect graphs with the same degree distribution to {\sl behave} similarly, e.g. we would anticipate an epidemy to spread on them in the same way. 

The degree distribution is however not the whole information about a graph. Below we introduce the notion of {\sl degree correlation matrix}, which comprises the probabilities that an edge points {\sl from} a vertex of degree $k$ {\sl to} a vertex of degree $k'$. Let us denote this variable by conditional probability $\mathbb{P}(k'|k)$. In the simplest case, e.g. for Erd\"os-Rényi random graphs (see Sec. \ref{er}) in the limit of size of a graph $|G|$ going to infinity, probabilities $\mathbb{P}(k'|k)$ do not depend on $k$ \cite{gnp-nocor}. We say then that such graphs exhibit \emph{no degree correlation}. As the exact formula for this conditional probability is going to be useful for us later on, we work out now this result.

Consider a graph $G=(E,V)$ of the size $N$ with no degree correlation and denote by $E_{kk'}$ the number of edges pointing from the set of vertices with degree $k$ to the set of vertices with degree $k'$. We stick to the case of undirected graph, and so $E_{kk'}$ is symmetric. The diagonal elements $E_{kk}$ will be thus equal to twice the number of connections within one class. The number of all edges departing from the set of vertices with degree $k$ is:
\begin{equation}
\sum_{k'}E_{kk'}=kN_k, 
\end{equation}
where $N_k$ is the total number of vertices of degree $k$. The conditional probability $\mathbb{P}(k'|k)$ is thus:
\begin{equation}
\mathbb{P}(k'|k)=\frac{E_{kk'}}{kN_k}.
\label{eq:a321}
\end{equation}
The sum of all the vertices' degrees is equal to twice the number of edges in the whole graph:
\begin{equation}
\sum_{k,k'}E_{kk'}=\sum_{k}k\frac{N_k}{N}N=\sum_{k}k\mathbb{P}(k)N=\mathbb{E}(k)N=2|E|. 
\label{eq:a123}
\end{equation}
The identity (\ref{eq:a123}) provides us the definition of the joint degree distribution:
\begin{equation}
\mathbb{P}(k,k')=\frac{E_{kk'}}{\mathbb{E}(k)N}.
\end{equation}
According to the last equation we transform (\ref{eq:a321}) into:
\begin{equation}
\mathbb{P}(k'|k)=\frac{\mathbb{E}(k)\mathbb{P}(k,k')}{k\mathbb{P}(k)}. 
\end{equation}
Finally, due to symmetry of $\mathbb{P}(k,k')$ in arguments, we get the so-called {\sl detailed balance condition}:
\begin{equation}
 k\mathbb{P}(k'|k)\mathbb{P}(k)=k'\mathbb{P}(k|k')\mathbb{P}(k').
 \label{eq:db}
\end{equation}
Now, as we know that for uncorrelated $\mathbb{P}(k'|k)$ does not depend on $k$, we may write $\sum_k k\mathbb{P}(k'|k)\mathbb{P}(k)=\mathbb{P}(k'|k)\mathbb{E}(k)$. Applying normalization condition $\sum_k \mathbb{P}(k|k')=1$ on the last equation gives:
\begin{equation}
\mathbb{P}(k'|k)=\frac{k'\mathbb{P}(k')}{\mathbb{E}(k)}. 
\label{eq:uncor}
\end{equation}
Note, that this simple formula was possible to get only thanks to the fact that $\mathbb{P}(k'|k)$ is independent of $k$. We obtained here the conditional probability $\mathbb{P}(k'|k)$ for the case of uncorrelated graph.

The type of the distribution $\mathbb{P}(k)$ that describes a network has a crucial impact on the dynamics. For practical issues associated with network analysis we distinguish two main network classes: {\sl heterogeneous} and {\sl homogeneous}, which roughly correspond to {\sl heavy-tailed} and, say, light-tailed degree distributions. 

A good example of homogeneous network is Erd\"os-Rényi random graph, described in the next section. In this type of graph the degree distribution is Bernoulli, or, in the limit of large size, Poisson distribution. In these distributions the average value represents, in a way, typical degree: values that are far away from the average are extremely rare.

This is, however, usually not the case in nature. Real networks often exhibit existence of so-called {\sl hubs}, i.e. nodes with exceptionally large number of neighbors. Their appereance can for instance manifest popularity of particular people (in social network) or be the result of optimization, as in the airport network. We have namely a few hub-airports like Atlanta, London, Frankfurt or Tokio, with dozens of connections and plenty of small airports with just a few connections. It works the same for scientific collaboration network \cite{NewGirv,simonsen}. Newman and Girvan \cite{NewGirv} created a network with nodes corresponding to scientists dealing with complex network problems, and put links between each two of them, who have at least one common paper on e-print archive arxiv.org . In Fig. \ref{fig:scinet} one can easily distinguish head-researchers, like Newman, Vespignani or Barab\'asi, who serve as hubs in this network. Moreover, what may be interesting for a Polish reader: in the right-up 
corner of the network we can observe the group of prof. Ho\l{}yst from Warsaw University of Technology.

\begin{figure}
\includegraphics[height=11cm]{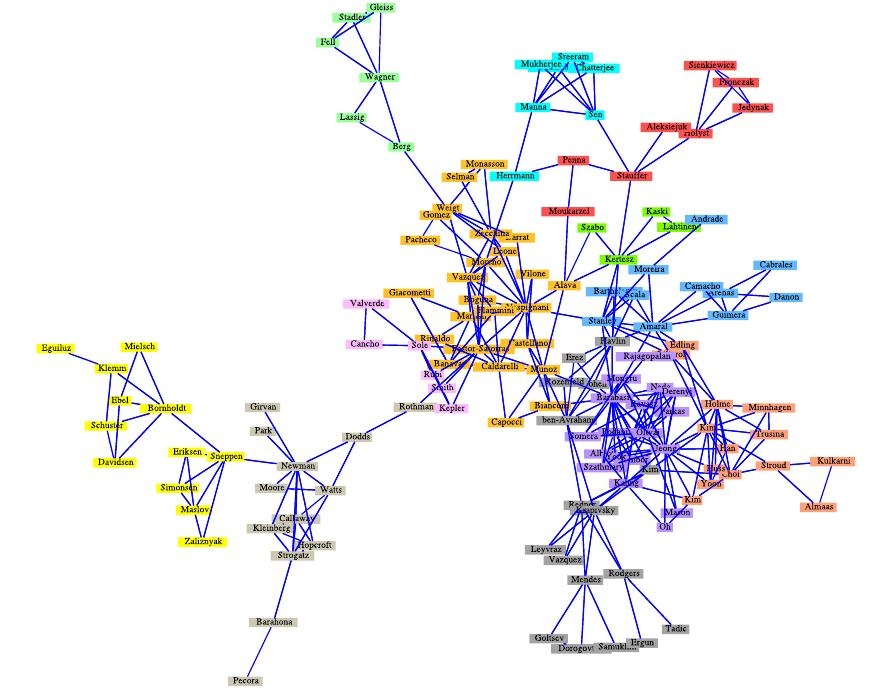}
 \caption{Scientific collaboration network: nodes represent researchers, links between two of them correspond to existence of their common paper. Picture taken from \cite{simonsen}.}
 \label{fig:scinet}
\end{figure}

The existence of highly connected nodes makes the type of the distribution significantly different form the homogeneous one: now the values which are far away from the average are non-negligible, and actually the average value does not tell us much about the distribution. Moreover, in many cases it causes heavy-tailed distributions, as described in \cite{pa}.
Intuitively, heavy-tails in distributions mean what we have just described: relatively (comparing to, for instance, normal distribution) large probabilities of values far from the average. Let us however introduce a strict definition of these objects \cite{heavy}.

\begin{df}
A distribution $F$ is said to be heavy-tailed if
$$\int_\mathbb{R} e^{\lambda x} F(dx)=\infty \textrm{ for all }\lambda>0.$$
Otherwise it is called light-tailed, i.e. if
$$\int_\mathbb{R} e^{\lambda x} F(dx)<\infty \textrm{ for some }\lambda>0.$$
\end{df}
So in particular for any light-tailed distribution $F$ on the positive half-line all moments are finite; 
$\int_0^\infty x^k F(dx)<\infty$ for all $k>0$.

Of course the strict definitions are in a sense useless in the real world problems: we will never get any infinity out of real data, as real network, although sometimes huge, are always finite. We have got, however, an intuition of what a heavy-tailed distribution is and how to examine our model of artificial networks. Prominent examples of heavy- and light-tailed distributions are power law and Poisson distributions respectively.

Coming back to the real world examples: take a look on Fig. \ref{fig:loglogi} where we put double logarithmic plots of degree distribution of four real networks: worldwide airport network (links correspond to connections between them), actors (links -- starring in the same movie), Autonomous Systems of the Internet and WWW ($k_in$ -- number of in-coming hyperlinks). Note, that degrees of nodes spans through several orders of magnitude. They, moreover, look linear at least in some regions, what suggest that they can be properly described by power law degree distribution: we are going to come back to this problem later.

\begin{figure}
  \includegraphics[height=8cm]{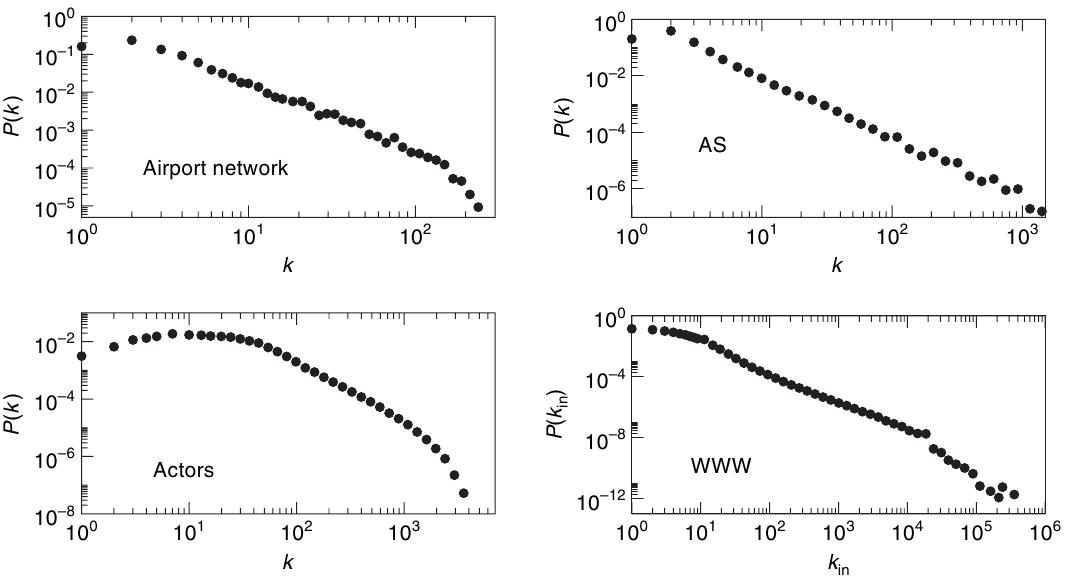}
  \caption{Degree distribution $\mathbb{P}(k)$ of four real-world networks in double logarithmic scale: worldwide airport
  network, actors, Autonomous Systems of the Internet and WWW. Picture taken from \cite{DynProc}.}
    \label{fig:loglogi}
\end{figure}

A heuristic way to distinguish heterogeneous and homogeneous type of degree distribution of a network is provided by {\sl heterogeneity parameter} defined as:
\begin{equation}
\kappa=\frac{\langle k^2 \rangle}{\langle k \rangle}, 
\label{eq:kappa}
\end{equation}
where $\langle . \rangle$ stands for an average value. For homogeneous networks we expect $\kappa\sim\langle k \rangle$, while for heterogeneous $\kappa\to\infty$. Again, in real world examples it is enough for ''detecting'' heterogeneity if $\kappa\gg\langle k \rangle$.

%można dodać sekcję o tym co to jest random graph, appendix 1 in vespignani

\section{Erd\"os-Rényi random graph}
\label{er}
The simplest model of social networks is the so-called Erd\"os-Rényi random graph \cite{ER}. Actually, we will be interested in a slightly different version of random graph proposed independently by Gilbert \cite{Gilbert} (nonetheless referred to as the Erd\"os-Rényi random graph). Although it appears not to be the most accurate model for real systems (see e.g. \cite{BarabasiRMP}), we will employ it as an example or special case. Conceptual simplicity of random graph allows often to achieve some exact solutions, not possible to get for general cases.\\  \indent 
As we will usually restrict considerations to connected graphs, besides the definition we will show that random graph are asymptotically almost surely connected. It was proven first by Erd\"os and Rényi \cite{ER2}, but we will follow the way presented in \cite{NOOC}. We also introduce the notion of properties that hold asymptotically almost surely.

\begin{df}
A random graph $G(n,p)$, where $n\in\mathbb{N}$, $p\in[0,1]$, is a graph with n vertices such that an edge between each pair of vertices exists independently with probability p.
\end{df}

\begin{df}
A property $\mathcal{P}$ holds asymptotically almost surely (a.a.s.) if the probability of this event goes to 1 when $n\to\infty$.
Then equivalently, for the case of graphs, we can say that almost every graph has property $\mathcal{P}$.
\end{df}

\begin{thm}[Almost every G(n,p) is connected]
For constant p a graph G(n,p) is connected a.a.s. .  
\end{thm}

\begin{proof}
%The fact that a graph is connected means also that we cannot decompose it into two disconnected graphs. 
%The idea of the proof is to show that the probability of obtaining such decomposition goes with to $0$ with $n\to\infty$.\\ \\
Assume $S,S'$ are two sets such that $V(G)=S\cup S'$. Let us fix $s=|S|$, then $|S'|=n-s$, thus:
$$\mathbb{P}\big(\{\textrm{S and S' are disconnected}\}\big)=(1-p)^{s(n-s)}.$$
We can upperbound the probability that $G(n,p)$ is disconnected by probabilities of existence of a partition of $V(G)$ into two subsets:
\begin{equation}
\mathbb{P}\big(\{\textrm{G(n,p) is disconnected}\}\big)=\mathbb{P}\big(\bigcup\limits_{S\cup S'=V(G)}\{\textrm{S and S' are disconnected}\}\big)\leq\sum\limits_{s=1}^{n/2}\binom{n}{s}(1-p)^{s(n-s)}.
\label{eq:con1}
\end{equation}
We sum here only up to $\frac{n}{2}$ as for $s$ greater than $\frac{n}{2}$ we would count the same partitions once again, but with roles of $S,S'$ exchanged. For binomial coefficients we have:
$$\binom{n}{s}=\frac{n!}{s!(n-s)!}=\frac{n\cdots(n-s+1)}{s!}\leq n^s.$$
Moreover, as $(1-p)<1$ and because we sum only up to $\frac{n}{2}$, we have $(1-p)^{n-s}\leq(1-p)^\frac{n}{2}$, hence:
\begin{equation}
\mathbb{P}\big(\{\textrm{G(n,p) is disconnected}\}\big)\leq\sum\limits_{s=1}^{n/2}\big(n(1-p)^{\frac{n}{2}}\big)^s.  
\label{eq:con2}
\end{equation}
Clearly for $n$ large enough we get $n(1-p)^\frac{n}{2}<1$, so we can further upper bound (\ref{eq:con2}) by a convergent geometric series:
\begin{equation}
\sum\limits_{s=1}^{n/2}\big(n(1-p)^\frac{n}{2}\big)^s\leq\sum\limits_{s=1}^{\infty}\big(n(1-p)^\frac{n}{2}\big)^s=\frac{n(1-p)^\frac{n}{2}}{1-n(1-p)^\frac{n}{2}}.
\label{eq:con3}
\end{equation}
Now we pass to the limit $n\to\infty$. Since $\lim\limits_{n\to\infty} n(1-p)^\frac{n}{2}=0$ we have:
\begin{equation}
\lim\limits_{n\to\infty}\frac{n(1-p)^\frac{n}{2}}{1-n(1-p)^\frac{n}{2}}=0,  
\end{equation}
thus consequently from equations (\ref{eq:con1})-(\ref{eq:con3}) we arrive with the result $\lim\limits_{n\to\infty}\mathbb{P}\big(\{\textrm{G(n,p) is disconnected}\}\big)=0$.

\end{proof}

In the following we would like to develop some characteristics of the $G(n,p)$ model that are interesting to compare with other network models and real-data graphs. Let us start with a simple calculation of mean number of edges and mean degree.

\begin{remark}
Expected value of mean number of edges $m$ in a random graph $G(n,p)$ is given by:
\begin{equation}
  \mathbb{E}(m)=\binom{n}{2} p.
\label{eq:meanm}
\end{equation}
\end{remark}
\begin{proof}
There are $\binom{n}{2}$ possible edges between $n$ vertices and each of them exists with the probability $p$. 
%But let us take a rigorous way of deriving $\mathbb{E}(m)$ following the approach presented in \cite{Networks}.
%\newline\newline
%In the ensemble of $G(n,p)$-graphs, each particular graph G appears with the same probability:
%$$\mathbb{P}(G)=p^m(1-p)^{\binom{n}{2}-m}.$$
%We want to know the total probability $\mathbb{P}(m)$ of obtaining a graph with $m$ edges. There are $\binom{\binom{n}{2}}{m}$ of them, %as this is the number of choosing $m$ edges from $\binom{n}{2}$ distinct pairs of certices, hence:
%$$\mathbb{P}(m)=\binom{\binom{n}{2}}{m}\mathbb{P}(G)=\binom{\binom{n}{2}}{m}p^m(1-p)^{\binom{n}{2}-m}.$$
%Now knowing the expected value of binomial distribution, see e.g. \cite{StochMod}, we finally obtain:
%$$\mathbb{E}(m)=\sum\limits_{m=0}^{\binom{n}{2}}m\mathbb{P}(m)=\binom{n}{2} p.$$
\end{proof}

Having in mind the latter result we can easily compute the mean degree of a $G(n,p)$-graph.
\begin{prop}
Expected value of the degree of a randomly chosen vertex in a random graph $G(n,p)$ is given by:
\begin{equation}
  \mathbb{E}(k)=(n-1)p.
\label{eq:meank}
\end{equation}
\end{prop}
\begin{proof}
The mean degree in a graph with $m$ edges is $\frac{2m}{n}$, as there are $2m$ edge stubs and $n$ vertices. 
Thus, using (\ref{eq:meanm}) and once again expected value of binomial distribution, we calculate:
$$\mathbb{E}(k)=\sum\limits_{m=0}^{\binom{n}{2}}\mathbb{E}(k\big| m)\mathbb{P}(m)=\sum\limits_{m=0}^{\binom{n}{2}}\frac{2m}{n}\mathbb{P}(m)=\frac{2}{n}\sum\limits_{m=0}^{\binom{n}{2}}m\mathbb{P}(m)=\frac{2}{n}\binom{n}{2}p=(n-1)p.$$
\end{proof}
Next property we are going to analyze is the degree distribution. This feature of $G(n,p)$ allows us to admit, that this is not the most accurate model for real-world networks, which have completely different distributions, see Sec. \ref{sf}.

\begin{prop}[$G(n,p)$ has a binomial degree distribution]
The probability for a vertex in $G(n,p)$ random graph to be connected with $k$ neighbors, where $0\leq k \leq n-1$, is:
\begin{equation}
p_k(n)=\binom{n-1}{k} p^k(1-p)^{n-1-k}.
\label{eq:dd} 
\end{equation}
\end{prop}
\begin{proof}
Each vertex of a $G(n,p)$-graph can be connected with any of $n-1$ other nodes independently with probability p.  
Hence, a particular vertex is connected to a particular set of $k$ vertices with probability $p^k(1-p)^{n-1-k}$.
Moreover, there are $\binom{n-1}{k}$ ways of picking $k$ out of $n-1$ vertices, what yields the statement (\ref{eq:dd}).
\end{proof}

Often, by network analysis, we are interested in asymptotical behaviour of graph's features in the limit $n\to\infty$.
In the next remark of this section we point out that the degree distribution of $G(n,p)$ goes to Poisson distribution in the limit $n\to\infty$ for fixed $(n-1)p$. This is a general feature of binomial distribution, see e.g. \cite{StochMod}, thus we omit the proof of this fact.

\begin{remark}[In the large $n$ limit $G(n,p)$ has a Poisson degree distribution]
The total probability for a vertex in $G(n,p)$ random graph to be connected with $k$ neighbors in the limit $n\to\infty$ for fixed $(n-1)p$ is:
\begin{equation}
\lim_{n\to\infty}p_k(n)= e^{-d}\frac{d^k}{k!},
\label{eq:ddn} 
\end{equation}
where we denote by $d=p(n-1)$ the expected degree of a vertex.
\end{remark}

An important feature which appears in the limit $n\to\infty$ and which becomes crucial while considering $G(n,p)$ as a model of real networks follows from the next theorem, proven e.g. in \cite{NOOC}.

\begin{thm}
For constant p, a $G(n,p)$-graph is connected and has diameter 2 a.a.s. .
\label{thm:diam2} 
\end{thm}

The last of the characteristics of $G(n,p)$ we are going to show is clustering coefficient $C$. This is straighforward though: in a random graph the probability that two vertices are connected is the same for all pairs and equals $p$, thus the probability that two friends of mine are also friends is $p$ as well.

\begin{remark}
For a $G(n,p)$ graph clustering coefficient $C=\frac{\mathbb{E}(k)}{n-1}$.  
\label{rem:clust}
\end{remark}

We already know that $\mathbb{E}(k)=p(n-1)$ (see (\ref{eq:meank})), so we could also just write $C=p$. However, the formula $C=\frac{\mathbb{E}(k)}{n-1}$ emphasises an important fact, that for fixed $\mathbb{E}(k)$ in the limit $n\to\infty$ clustering coefficient for $G(n,p)$ tends to zero. 
%We will see soon that it is going to be an important fact while comparing $G(n,p)$ with real systems.

\section{Small worlds} 
\textit{It is a small world!} -- this phrase is not only a popular saying. Stanley Milgram, an experimental psychologist, showed in his famous experiment performed in the 1960s \cite{Mil1, Mil2}, how small our world actually is. Milgram sent 96 packages to random recipients in Omaha (Nebraska, US) with the instruction to attempt to send it forward to the particular person in Boston, a city situated over a thousand miles away from Omaha. The only way Milgram allowed them to do it was to transfer the parcel to the person who -- in their opinion -- could have bigger chance to know that guy from Boston. Eighteen parcels arrived to the final recipient -- that is already a remarkable achievement. Moreover, it came out that the average lenght of the completed paths was as small as 5.9 steps. This significant result yielded the popular "six degrees of separation" -- the idea that there are only six steps between us and anyone in the world. This is also the property we would like to mimic in mathematical models of 
social 
networks.

Various analysis of real networks (the world wide web, ecological networks, social networks etc., see e.g. \cite{Networks}) show that they character is defined by two main features: small world property and large clustering coefficient. From Theorem \ref{thm:diam2} we see that -- at least for graphs big enough -- we keep this property in $G(n,p)$ model. This is though not the case for clustering coefficient, which falls with $n$ to 0, as follows from the Remark \ref{rem:clust}. In 1998 Watts and Strogatz proposed \cite{SW} a solution for this problem: {\sl the small world model}, which captures both real networks features. %: small diameter and large clustering coefficient.
In fact the orginal model \cite{SW} yields many problems for rigorous analysis. Therefore we will -- as it is usually done in literature -- focus on a slightly different model proposed by Newman and Watts in \cite{SW2}. The second model captures all the features important for imitating real networks. The difference between two versions is to be indicated later on. 

\begin{figure}
  \includegraphics[height=7cm]{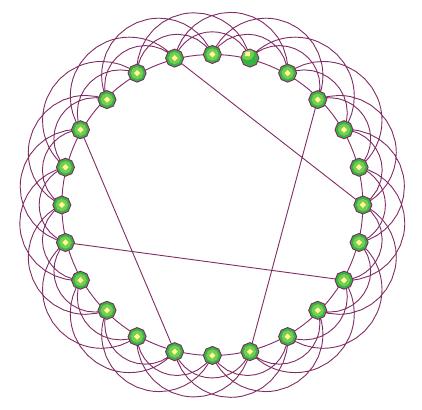}
  \caption{Small world model proposed by Newman and Watts in \cite{SW2}. Here the number of nodes n=24 and parameter c=6. Picture taken from \cite{SW3}.}
   \label{fig:sw}
\end{figure}

The idea is to perform an interpolation between random graph and a network with high clustering coefficient. To do so we start with a circle of connected nodes (see Fig. \ref{fig:sw}). Each node on the circle has a common edge not only with two neighbors, but with $c$ closest nodes in symmetric way ($c/2$ both on the left and right, $c$ is assumed to be an even integer). In this way we imitate the first property of real networks. Suppose this is the end of our construction and let us compute the clustering coefficient without \emph{shortcuts} (we will specify what are they in a moment). 

\begin{prop}
The clustering coefficient for the small world model without shortcuts is $C=\frac{3(c-2)}{4(c-1)}$.    
\end{prop}
\begin{proof}
We will compute the clustering coefficient according to definition \ref{df:cc2}, i.e. we need to count all triangles and paths of length 2 (2-paths) in the system. Any triangle has to contain two hops in the same direction around the circle and one in the opposite. However this last step back can be at most $\frac{c}{2}$ long as this is the length of the longest link in the network we consider. There are thus $\binom{c/2}{2}$ ways of reaching the futhermost vertex in two steps and summing up there are $n\binom{c/2}{2}=\frac{1}{4}nc(\frac{1}{2}c-1)$ triangles.
In order to count all the 2-paths in the system let us first focus on the paths centered on a particular vertex. Each vertex has $c$ neighbors and thus $2\binom{c}{2}$ possibilities of forming a 2-path, where the factor 2 stands for two possible directions. Therefore there are $nc(c-1)$ 2-paths. Now we easily get the final result:
$$C=\frac{\frac{1}{4}nc(\frac{1}{2}c-1)\times 6}{nc(c-1)}=\frac{3(c-2)}{4(c-1)}$$
\end{proof}

\begin{ex} Let us consider a few particular examples with differend value of parameter $c$:
\begin{itemize}
  \item $C(c=2)=0$ (a circle without extra neighbors), 
  \item $C(c=6)=0.6$,
  \item $\lim\limits_{c\to\infty}C(c)=\frac{3}{4}$.
\end{itemize}
\end{ex}

The conclusion is that we defined model with large clustering coefficient. Moreover: we do not have to go with parameter $c$ to $\infty$ to obtain high values of $C$: it is enough to add just a few extra neighbors on the circle. But the problem is that we have now a kind of \textit{large} world: the mean shortest path is as big as $\frac{n}{2c}$ \cite{Networks} (think of a way from a particular node to the one situated in the opposite place on the circle). In order to fix it we add the \textit{shortcuts} mentioned above. In the original version of small-world model \cite{SW} each edge we rewire with some probability $p$ such that it joins two vertices chosen uniformly at random. But as this formulation of the model yields some severe problems while analizing it \cite{SW2}, the version which is actually in use differs a bit. Namely, instead of \emph{rewiring} edges, we add some extra connections: with probability $p$ we put a link between two vertices chosen uniformly at random. In order to save compatibility 
with the results for original smal-worlds model \cite{SW}, we perform this procedure once for each edge from the circle (although we do not touch those edges!). To ilustrate, that this rewiring indeed brings us to the regime of small diameters, we cite after \cite{blb} the following theorem:

\begin{thm}
Let G be an undirected graph formed by adding a random matching to an n-cycle. Then G has diameter diam(G) satisfying a.a.s.:
$$\log_2 n \leq diam(G) \leq \log_2 n + \log_2 \log n + b, $$
where b is a constant (at most 10).
 \label{thm:diam-sw} 
\end{thm}

There exist a broad regime when high clustering and small average distance are present simultaneously. It can be observed in Fig. \ref{fig:sw-plot}.

\begin{figure}
  \includegraphics[height=8cm]{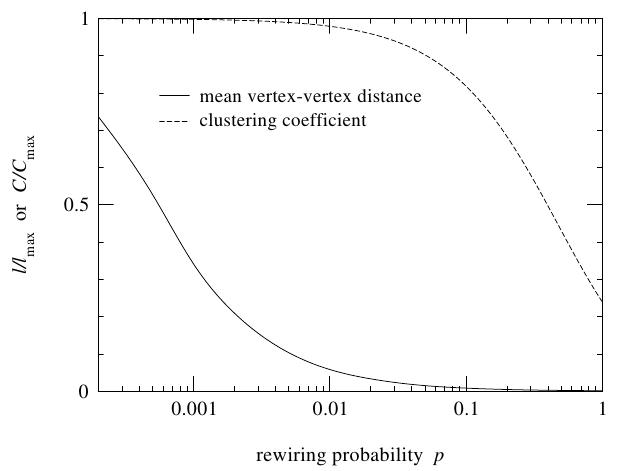}
  \caption{Average distance $l$ and clustering coefficient $C$ dependence on rewiring probability $p$. Picture taken from \cite{NewmanSIAM}.}
  \label{fig:sw-plot}
  \end{figure}

\section{Scale-free networks}
\label{sf}
The third model of networks that we present are so-called scale-free networks. Although the concept of \textit{scalefreeness} is a little bit wider, it practically boils down to the power law degree distribution graph. We will be thus speaking about graphs with degree distribution $\mathbb{P}(k)$ given by:
$$\mathbb{P}(k)=\left\{ \begin{array}{ll}
                   0 & \textrm{for } k=0, \\
		   k^{-\alpha}/\zeta(\alpha) & \textrm{for } k>0, \\ 
		    \end{array} \right.$$
Here $\zeta(\alpha)$ stands for the Riemann zeta function, which normalizes the distribution. Note that a power law does not posses a characteristic scale. Indeed, in the {\sl heavy-tailed} power law distribution defined above, values that lie far away from the mean (highly connected {\sl hubs} mentioned above) are much more probably than for {\sl light-tailed} distributions, such as the exponential. In other words: in {\sl light-tailed} distributions the average value gives some kind of scale, but it is not the case anymore for power law, which is a {\sl heavy-tailed} distribution.

This kind of behavior of a network appears to be ubiquitous in nature. Power law degree distributions can be found in as different contexts as World Wide Web \cite{www}, network of citations \cite{citations}, protein interactions \cite{protein}, movie actors \cite{film} or sexual contacts \cite{sex}. Not only real degree distributions are characterized by power laws, we can encounter them also analyzing earthquakes \cite{quakes}, the number of species in biological taxa \cite{taxa}, the frequency of use of words in human languages \cite{word}, GDP of The United Kingdom \cite{economics} and much more \cite{PowerLawReview}. The main class of mechanisms that cause this kind of degree distribution has been discovered by Albert and Barab\'asi \cite{pa} and is usually referred to as {\sl the rich get richer} or {\sl preferential attachment}. Strictly speaking, Albert and Barab\'asi \emph{rediscovered} this mechanism after Yule \cite{yule} and Simon \cite{simon}. In order to explain the mechanism, let us 
think of populations of cities. Consider a person that is about to move to another place to live. There are many reasons why he or she would go here or there, but in general people try to choose the most comfortable place. If a particular place really is comfortable, we would expect many people 
already living there. Indeed, big cities attract people because, for instance, it is usually easier to find a job there. Therefore we can guess the following rule: the more populated a city is, the more eager people are to move there. And this basically is the idea of {\sl preferential attachment}. It seems to be very natural, so we should not be that much surprised by power law behavior appearing everywhere. {\sl The rich get richer} is not, however, the only mechanism that produces power law, see e.g. \cite{NOOC}.

In real-world realizations of networks power law usually applies to the tail of a particular distribution. We fix then the minimum value $k_{min}$ over which the power law holds: $\mathbb{P}(k)=Ck^{-\alpha}$ for $k\geq k_{min}$, and normalize only the tail \cite{Networks}. Using the fact that the power law is a slowly varying function of $k$ for large $k$ we approximate the sum by an integral:
\begin{equation}
 C=\frac{1}{\sum_{k'=k_{min}}^{\infty}k^{-\alpha}}\approx\frac{1}{\int_{k=k_{min}}^{\infty}k^{-\alpha}dk}=(\alpha-1)k_{min}^{\alpha-1},
 \label{eq:trick}
\end{equation}
where we assume $\alpha> 1$ so that the integral converges. We can then approximate the distribution by:
\begin{equation}
 p_k\approx\frac{\alpha-1}{k_{min}}\Big(\frac{k}{k_{min}} \Big)^{-\alpha}.
\end{equation}
Let us compute now $m$th moment $\langle k^m \rangle$ of such a distribution, i.e. distribution that follows power law beyond some $k_{min}$:
$$\langle k^m \rangle=\sum_{k=0}^{k_{min}-1}k^m p_k + C \sum_{k=k_{min}}^{\infty}k^{m-\alpha}.$$
We will use again the same trick as in eq. (\ref{eq:trick}): approximation of the sum by an integral:
\begin{equation}
\langle k^m \rangle\approx  \sum_{k=0}^{k_{min}-1}k^m p_k + C\int_{k=k_{min}}^{\infty}k^{m-\alpha}dk=
\sum_{k=0}^{k_{min}-1}k^m p_k +\frac{C}{m-\alpha+1} \Big[k^{m-\alpha+1} \Big]^{\infty}_{k_{min}}.
\end{equation}
There are no problem with the first sum, as it is a finite number. The second one, however, could be infinite, depending on the values of $\alpha$ and $m$: given value of $\alpha$ all moments greater or equal than $\alpha-1$ diverge. In particular, the second moment is finite if and only if $\alpha>3$. In the same time, the real-world networks following power law are usually characterized by $2\leq \alpha \leq 3$, so consequently, the second moment should diverge. This conclusion is slightly misleading though: in fact, for any finite system -- and real systems obviously are finite -- any finite moment has finite value. The last result, however, is crucial for epidemic spreading predictions for models described in the next chapter.

Apart from the hub emergence, which makes scale-free networks similar to real systems, they capture also two other important features already mentioned above, which are usually present in nature: the small-world property and clustering. Indeed, the average distance for network with power law degree distribution with the exponent $2 < \alpha < 3$ is almost surely of the order $\log \log n$ \cite{sf-d}, while the clustering coefficient grows with $n$ to infinity \cite{Networks}. This divergence happens of course only in theoretical models in the large $n$ limit and with an assumption that the network is created at random. In reality, first of all: size $n$ of the network can be large, but not infinite, but moreover: real networks are not purely random. There usually appear some {\sl motifs}, i.e. subgraphs with a specific structure, which are for some reason overrepresented in a particular network.

\chapter{Stochastic processes}
\label{markov}
In this chapter we introduce some rudiments of Markov chains, as all the models of epidemic spreading which we describe further on, are of the Markovian character. The {\sl Markovian} property intuitively means, that the process being analyzed has no memory: in order to perform a time step, the system needs to know only the present state. Quoting R.A.~Howard: {\sl Markov chain is like a frog jumping on a set of lily pads} \cite{frog}. In mathematical modeling the {\sl Markovian} property is an assumption which simplifies the computations significantly, and usually does not take us too far from reality. 

The simpliest, and at the same time famous example of a Markov chain is one dimensional random walk: a model for Brownian motion, first proposed by M.~Smoluchowski in the beginning of 20th centrury \cite{smoluch}. In this process we move on the set of integer numbers. At each time step we go one step left or right, according to the hopping probabilities, so the position in the next time step depends only on the present position and these probabilities.

\section{Elements of Stochastic Processes}
First we introduce the notion of stochastic process and related concepts. Let $(\Omega,\Sigma,\mathbb{P})$ be a probability space and $T$ be a countable non-empty set. We restrict ourselves to discrete-time stochastic processes with finite state spaces. In this introduction we follow \cite{bremaud,koralov,shah,stir}.

\begin{df}
A stochastic process is a map $X:T\times\Omega\mapsto \mathbb{R}^n$ such that for every $t\in T$ the transformation: 
$$ \Omega\ni\omega\mapsto X(t,\omega)$$ 
is a random variable. Alternatively, we can just enumerate time $t$ by natural numbers $n\in\mathbb{N}$ and represent a stochastic process by a sequence $\{X_n(\omega)\}_{n\geq0}$. 
\end{df}
One of the specific kinds of stochastic processes are martingales. They may be thought of as models for fair games, i.e. your expected fortune after any bet is equal to your fortune before. This intuitive feature is expressed in the third point of the following definition.

\begin{df}
A real-valued stochastic process $\{Y_n(\omega)\}_{n\geq0}$ is called a martingale with respect to a stochastic process $\{X_n(\omega)\}_{n\geq0}$ if for each $n\geq 0$:
\begin{enumerate}
 \item $Y_n$ is a function of $X_0,...,\,X_n$, 
 \item $\mathbb{E}(|Y_n|)<\infty\textrm{ or }Y_n\geq 0$, and
 \item $\mathbb{E}(Y_{n+1}|X_0,...,\,X_n)=Y_n$.
\end{enumerate}

\end{df}

One however would admit, that epidemy does not look really fair. Indeed, the notion that we will use later one is {\sl supermartingale}, which can be thought of as a game in which you loose, i.e. your expected fortune is decreasing in time.
For completeness we add also the definition of submartingale.

\begin{df}
A real-valued stochastic process $\{Y_n(\omega)\}_{n\geq0}$ is called supermartingale (submartingale) with respect to a stochastic process $\{X_n(\omega)\}_{n\geq0}$ if it fulfills conditions $1.$ and $2.$ for martingales and moreover:
$$\mathbb{E}(Y_{n+1}|X_0,...,\,X_n)\leq\big(\geq\big)Y_n.$$
\label{df:sup}
\end{df}

Before we enter the land of Markov chains, we prove one fact that does not require {\sl Markovianity}, but will be useful for us later on.
%$(\Omega, \mathcal{F},\mathbb{P})$
\begin{thm}[Markov's inequality]
Let $X$ be a random variable on a discrete probability space $(\Omega,\Sigma,\mathbb{P})$. If $X\geq 0$ and $\mathbb{E}(X)$ is finite, then for each $a>0$ the following inequality applies:
$$\mathbb{P}(X\geq a)\leq \frac{\mathbb{E}(X)}{a}$$.
\label{thm:me}
\end{thm}
\begin{proof}
This statement can be proven by a sequence of inequalities:
$$\mathbb{P}(X\geq a)=\sum_{\{\omega:X(\omega)\geq a \}}\mathbb{P}(\omega)\leq\sum_{\{\omega:X(\omega)\geq a \}}\frac{X(\omega)}{a}\mathbb{P}(\omega)=\frac{1}{a}\sum_{\{\omega:X(\omega)\geq a \}}X(\omega)\mathbb{P}(\omega)\leq
\frac{1}{a}\sum_{\omega\in\Omega}X(\omega)\mathbb{P}(\omega)=\frac{\mathbb{E}(X)}{a}.$$
\end{proof}

\section{Markov chains}
From now on we restrict our considerations to a special case of stochastic processes, namely Markov chains.
\begin{df}
Let $\{X_i\}_{i\geq0}$ be a discrete-time stochastic process with finite state space E. This stochastic process is called {\sl Markov chain} if for all $n\geq 0$ and $i_0, i_1,...,i_{n-1},i,j\in E$ there holds:
\begin{equation}
\mathbb{P}(X_{n+1}=j|X_n=i,\, X_{n-1}=i_{n-1},...,\,X_{0}=i_{0})=\mathbb{P}(X_{n+1}=j|X_n=i),
\label{eq:marpro}
\end{equation}
whenever both sides are well-defined.
\end{df}
The property (\ref{eq:marpro}) is called the {\sl Markov property} and correspond to the fact that the process has no memory. It is already a big simplification, but we introduce yet another constraint:
\begin{df}
A stochastic process that fulfils (\ref{eq:marpro}) with the right hand side independent of $n$ is called a {\sl homogeneous} Markov chain (HMC).
\end{df}

Homogeneity of a Markov chain means that the probability of transition from state $i$ to state $j$ does not depend on time. These transition probabilities we usually put together into one object, called {\sl transition matrix}.

\begin{df}
The matrix $P=[p_{ij}]_{n\times n}$ where $i$ and $j$ are elements of the state space and $p_{ij}=\mathbb{P}(X_{k+1}=j|X_k=i)$ is called the transition matrix of the HMC.
\end{df}

This transition matrix cannot be just any matrix we want. We introduce the notion of {\sl stochastic matrix}, as further on we are going to deal with such matrices only.

\begin{df} A $n\times n$ matrix $A=[a_{ij}]_{n\times n}$ is called stochastic if:
 \begin{enumerate}
  \item $a_{ij}\geq 0$,
  \item $\sum^n_{j=1}a_{ij}=1$ for any $1\leq i \leq n$.
 \end{enumerate}
 If we substitute the second condition by $\sum^n_{j=1}a_{ij}\leq1$ for any $1\leq i \leq n$, then $A$ is called substochastic.
\end{df}

Directly from the construction and the normalization of probability (a process starting from a state $i$ in the next step has to go {\sl somewhere}, i.e. $\sum_{j\in E}p_{ij}=1$), the following property arises:

\begin{prop}
The transition matrix of HMC is a stochastic matrix. 
\end{prop}

The definition of a stochastic matrix can be expressed in other ways. The one particularly useful for us is pointed out in the following lemma.
\begin{lemma}
The two statements are equivalent:
\begin{enumerate}[(i)]
 \item A $n \times n$ matrix $A$ is stochastic.
 \item If $\mu=(\mu_1,...,\mu_n)$ is a probability vector, that is $\mu_i\geq 0$ and $\sum^n_{i=1}\mu_{i}=1$, then $\mu A$ is also a probability vector.
\end{enumerate}
\label{lemma:et}
\end{lemma}
\begin{proof}
First we show that (i) implies (ii). Denote $\nu=\mu A$, so $\nu_j=\sum^n_{i=1}\mu_i a_{ij}$. From stochasticity of A we conclude that $\nu_j \geq 0$. To complete the proof of the fact that $\nu$ is a probability vector, we calculate:
$$\sum_{j=1}^n\nu_j=\sum_{j=1}^n\sum_{i=1}^n\mu_i a_{ij}=\sum_{i=1}^n\sum_{j=1}^n\mu_i a_{ij}=\sum_{i=1}^n\mu_i=1,$$
where in the last equality we used the fact that $\mu$ is a probability vector, and in the penultimate the stochasticity of A.

Now we conclude (i) from (ii). Denote $\delta^{(i)}=(0,...,0,1,0,...0)^T$, where $1$ stands at the $j$-th entry. $\delta^{(i)}$ is a probability vector. From (ii) we know, that $\delta^{(i)} A=\{ \sum_{j=1}^n \delta^{(i)}_j a_{jk}\}_{k=1,...,n}=\{a_{ik}\}_{k=1,...,n}$ is a probability vector too, so for all $i,k$ there is $,a_{ik}\geq 0$ and for all $i$ we have $\sum_{k=1}^n a_{ik}=1$, what completes the proof.
\end{proof}

Transition matrix, together with some initial conditions, determines the dynamics of the process. Indeed, any probability describing HMC can be represented by the means of its transition matrix.

\begin{remark}
The distribution of a discrete-time homogenous Markov chain is determined by its initial distribution and its transition matrix. In patricular the following equality holds: 
$$\mathbb{P}(X_{n+1}=j,X_n=i,\, X_{n-1}=i_{n-1},...,\,X_{0}=i_{0})=p_{n,n-1}p_{n-1,n-2},...,p_{1,0}p_0,$$
where $p_0=\mathbb{P}(X_0=i_0)$.
\end{remark}
\begin{proof}
From Bayes' rule there follows:
$$\mathbb{P}(X_{n+1}=j,X_n=i,\, X_{n-1}=i_{n-1},...,\,X_{0}=i_{0})=\mathbb{P}(X_{n+1}=j|X_n=i,\, X_{n-1}=i_{n-1},...,\,X_{0}=i_{0})\times$$
$$\times \mathbb{P}(X_n=i,\, X_{n-1}=i_{n-1},...,\,X_{0}=i_{0}).$$
From the Markov property (\ref{eq:marpro}) we know that the first term on the right hand side is nothing but $\mathbb{P}(X_{n+1}=j|X_n=i)$. We follow this procedure step by step and finally we arrive at:
$$\mathbb{P}(X_{n+1}=j,X_n=i,\, X_{n-1}=i_{n-1},...,\,X_{0}=i_{0})= \mathbb{P}(X_{n+1}=j|X_n=i),...,\,
\mathbb{P}(X_{1}=i_1|X_0=i_0)\mathbb{P}(X_0=i_0),$$
which is what we wanted to show.
\end{proof}

If we need to perform several transitions from some initial distribution, the notion of {\sl n-step transition matrix} will be useful.

\begin{df}
The matrix $P^n$ is called n-step transition matrix and:
$$p_{ij}(n)= \mathbb{P}(X_{m+n}=j|X_m=i)=\sum_{i_1,...,i_{n-1}}p_{i,i_1}p_{i_1,i_2},...,p_{i_{n-1},j}.$$
\end{df}
Proof of correctness of this definition can be found e.g. in \cite{koralov}.
\newline

Practically, given a vector of probabilities at some time $t$: $q(t)_i=\mathbb{P}(X_t=i)$, if we would like to provide a vector of probabilities for the next time step $t+1$, we simply multiply vector by the transition matrix of a given HMC: 
$q(t+1)=\sum_{i\in E}p_{ij}q_i(t)=q(t) P$. If we are to perform $n$ time steps, we multiply n times: $q(t+n)=q(t) P^n$. Note, that with the construction presented here, we have to multiply matrix by {\sl horizontal} vector from the left hand side. This is of course just a matter of notation convention: we could have defined the transition matrix in the way that $p_{ji}$ is the probability of choosing a state $j$ starting from a state $i$ and then we would have had multiplication of matrices by {\sl vertical} vectors from the right hand side, which is usually considered as more intuitive in the algebraic computations. We traded this intuition for another: in the matrix element $p_{ij}$ the first index $i$ enumerates a place where we are, and the second $j$ -- where we can go -- which is at least chronological. For convinience we adopt a convention that vector $q$ is horizontal, while transposed vector $q^T$ is vertical.

%wrzucić tu jeszcze Chapmana - Kołmogorowa?

There is a deep analogy between Markov chains and graph theory. We can think of states of a Markov chain as nodes of a graph. We draw this analogy by providing yet a bit more definitions and then describing their relations to the graph theory, see Sec. \ref{graph}.

\begin{df}
State $j$ is said to be accesible from state $i$ if there exists $t\geq 0$ such that $p_{ij}(t)>0$ 
\end{df}
\begin{df}
States $i$ and $j$ are said to communicate if $i$ is accessible from $j$ and $j$ is accessible from $i$.
\end{df}
\begin{df}
For any HMC with the state space $E$ if any state $i$ from a given set of states $S\subseteq E$ communicates with any other state $j$ from this set, then $S$ is said to be irreducible. If $S=E$, then HMC is said to be irreducible.
\end{df}
\begin{df}
The period $d(j)$ of a state $j$ is the greatest common divisor of the times at which a return to this state is possible. State $j$ is aperiodic if $d(j)=1$, otherwise it is periodic. If a HMC contains a periodic state, we say that this HMC has periodic structure.   
\end{df}

All this we have seen already on graphs: accessibility of state $j$ from state $i$ corresponds to one way connectivity from node $i$ to node $j$ in directed graph. Two communicating states $i,j$ correspond to two connected nodes in a graph. Irreducible HMC straightforwardly corresponds to strongly connected graph. Adjacency matrix which is stochastic, can serve as a transition matrix for a HMC. So in order to analyze HMC we can equivalently take the matrix that describes it or the corresponding weighted graph. Note, that any adjacency matrix can be made stochastic by normalization: a matrix with elements $a'_{ij}=a_{ij}/\sum_{j} a_{ij}$ is already stochastic. Period of a HMC has also its natural analogue on graphs.

\section{Ergodic theorem}
In this section we present two versions of ergodic theorem for Markov chains. There are some differences between them, but both deal with the stationary state of a process. This is very important in our case of epidemic spreading: we will be interested in number of infected individuals after some long times.

\begin{df}
A probability distribution $\pi$ on a finite state space $E$ is said to be stationary for a transition matrix $P$ if:
$$\pi P=\pi.$$
\end{df}

The first version of ergodic theorem (to be found e.g. in \cite{koralov}) that we present simply points what kind of Markov chain can actually achieve stationary distribution. Let us firts define this special class.

\begin{df}
A Markov chain is said to be ergodic if there exists $n$ such that the {\sl n-step transition probabilities} $p_{ij}(n)$ are strictly positive for all indices $i,j$. 
\end{df}

So basically the idea is to be able to reach any state $j$ starting from any other state $i$ after some time. This is not always the case: let us think about a Markov chain described by a block-diagonal transition matrix -- such a Markov chain is obviously not ergodic.

\begin{thm}[Ergodic theorem 1]
Given an ergodic homogeneous Markov chain with $n$-dimensional state space and the transition matrix $P$, there exists a unique stationary probability distribution $\pi$. For every $i$ the t-step transition probabilities converge to the distribution $\pi$:
$$\lim_{t\to \infty}p_{ij}(t)=\pi_j.$$
%The stationary distribution satisfies $\pi_j>0$ for $1\leq j \leq n$.
\label{thm:ergodic}
\end{thm}
\begin{proof}
 The idea of this proof is first to show two inequalities between some distances on the space of probability distributions. We use them to show that any probability distribution transformed $t$-times by transition matrix $P$ forms a Cauchy sequence enumerated by $t$, what will allow us to show the convergence of $p_{ij}(t)$.
 
 Let $\mu,\nu$ be two probability distributions on $E=\{1,...,n\}$. We denote by $d$ a distance on the space of probability distributions:
 $$ d(\mu,\nu)=\frac{1}{2}\sum_{i=1}^n|\mu_i-\nu_i|. $$
 Clearly, $d(\mu,\nu)\leq 1$ .The space of distributions with such a distance is known to be a complete metric space.
 
 It is convinient to rewrite the distance in the following way:
 $$d(\mu,\nu)=\frac{1}{2}\sum_{i=1}^n(\mu_i-\nu_i)\Theta(\mu_i-\nu_i)+\frac{1}{2}\sum_{i=1}^n(\nu_i-\mu_i)\Theta(\nu_i-\mu_i)=
 \sum_{i=1}^n(\mu_i-\nu_i)\Theta(\mu_i-\nu_i),$$
 where $\Theta(x)$ stands for the Heaviside's step function which takes the value $1$ for $x\geq 0$ and $0$ otherwise. In the last equality we used the fact that for probability distributions $\mu,\nu$ we have $\sum_{i=1}^n \mu_i=\sum_{i=1}^n \nu_i=1$, and thus:
 $$\frac{1}{2}\sum_{i=1}^n(\nu_i-\mu_i)\Theta(\nu_i-\mu_i)=\frac{1}{2}\sum_{i=1}^n(\nu_i-\mu_i)-\frac{1}{2}\sum_{i=1}^n(\nu_i-\mu_i)\Theta(\mu_i-\nu_i)=\frac{1}{2}\sum_{i=1}^n(\mu_i-\nu_i)\Theta(\mu_i-\nu_i).$$
 
 Take now $A$ -- stochastic matrix. By Lemma \ref{lemma:et} be conclude that $\mu A$ and $\nu A$ are probability distributions.
 Below we are going to show the two following inequalities:
 \begin{eqnarray}
 d(\mu A, \nu A)&\leq & d(\mu,\nu), \label{eq:et1} \\ 
 d(\mu A, \nu A)&\leq & (1-\alpha)d(\mu,\nu), \label{eq:et2} 
 \end{eqnarray}
 where the latter inequality applies in case when for some $\alpha$ all $a_{ij}\geq\alpha$. Let us first estimate $d(\mu A, \nu A)$ in the general case:
 
 $$ d(\mu A, \nu A) = \sum_{j=1}^n (\mu A-\nu A)_j \Theta((\mu A-\nu A)_j) = 
 \sum_{j=1}^n \sum_{i=1}^n (\mu_i -\nu_i)A_{ij} \Theta((\mu A-\nu A)_j)\leq $$
 $$\leq\sum_{i=1}^n (\mu_i -\nu_i)\Theta(\mu_i -\nu_i) \sum_{j=1}^n A_{ij} \Theta((\mu A-\nu A)_j)\leq
 \sum_{i=1}^n (\mu_i -\nu_i)\Theta(\mu_i -\nu_i)=d(\mu,\nu),$$
 
 \noindent which gives us Eq. (\ref{eq:et1}). Assume now, that all $a_{ij}>\alpha$. Since both $\mu A$ and $\nu A$ are probability distributions, there must be at least one index $k$ such that $(\mu A)\leq(\nu A)$, and thus for all $i$ there applies $\sum_{j=1}^n a_{ij} \Theta((\mu A-\nu A)_j)<1-\alpha$. Modifying our previous estimations by introduction of the latter result we prove Eq. (\ref{eq:et2}):
 $$ d(\mu A, \nu A) \leq (1-\alpha)\sum_{i=1}^n (\mu_i -\nu_i)\Theta(\mu_i -\nu_i)=(1-\alpha)d(\mu,\nu).$$
 
 Now we are ready to show the convergence of $t$-step transition probability. Denote by $\mu_0$ an arbitrary probability distribution on Markov chain state space with positive entries and moreover: $\mu_t=\mu_0 P^t$. Using Eq. (\ref{eq:et1}-\ref{eq:et2}) we show that $\mu_t$ is a Cauchy sequence:
 
 $$ d(\mu_t,\mu_{t+k})=d(\mu_0 P^t,\mu_0 P^{t+k})\leq(1-\alpha)d(\mu_0 P^{t-s},\mu_0 P^{t+k-s})\leq $$
 $$ \leq ... \leq(1-\alpha)^m d(\mu_0 P^{t-ms},\mu_0 P^{t+k-ms})\leq (1-\alpha)^m,$$

 \noindent where we used the fact that $d(\mu,\nu)\leq 1$ and where we go with $m$ as far as $0\leq t-ms<s$ is fulfilled. For large $t$ the factor $(1-\alpha)^m$ can be arbitrarily small, and so we showed that $\mu_t$ is a Cauchu sequence. From completeness of the space, $\mu_t$ has a limit which is a probability distribution, let us denote it by $\pi=\lim_{t\to\infty}\mu_t$, and indeed, the probability distribution $\pi$ is stationary:
 $$\pi P= \lim_{t\to\infty}\mu_t P= \lim_{t\to\infty}(\mu_0 P^t )P =\lim_{t\to\infty}(\mu_0 P^t+1)=\pi.$$
 Let us take now a probability distribution $\mu_0$ to be concentrated at the point $i$, i.e. $\mu_0=(0,...,0,1,0,...0)^T$ where $1$ stands at the $i$-th entry. Then $\sum_{k=1}^n(\mu_0)_k (P^t)_{kj}=p_{ij}(t)$, so finally we proved that $\lim_{t\to \infty}p_{ij}(t)=\pi_j$. 
 %The feature: $\pi_j>0$ for $1\leq j \leq n$ follows from the definition of ergodic Markov chain.
 
 The last thing we have to show is uniqueness. Assume that there are two stationary distributions $\pi_1$ and $\pi_2$, i.e. $\pi_1=\pi_1 P$ and $\pi_2=\pi_2 P$. Then also $\pi_1=\pi_1 P^s$ and $\pi_2=\pi_2 P^s$. However, by Eq. (\ref{eq:et2}) we have:
 $$ d(\pi_1,\pi_2)=d(\pi_1 P^s,\pi_2 P^s ) \leq (1-\alpha)d(\pi_1,\pi_2),$$ which implies that $d(\pi_1,\pi_2)=0$, and thus $\pi_1=\pi_2$.
 
\end{proof}

This theorem will practically allow us to expect a stationary distribution in our considerations of epidemic spreading. But apart from theoretical analysis of epidemic spreading models we are also checking our predictions by computer simulations. The second version of ergodic theorem says that we can actually obtain something reasonable out of those simulations. For finite state spaces, which is the case we analize, it can be formulated as follows:

\begin{thm}[Ergodic theorem 2]
Let $\{X_n\}_{n\geq0}$ be an irreducible homogeneous Markov chain with state space $E$ and the stationary distribution $\pi$, and let $f:\, E\to\mathbb{R}$ be such that $\sum_{i\in E}|f(i)|\pi(i)<\infty$. Then for any initial distribution:
$$\lim_{n\to\infty}\frac{1}{N}\sum_{k=1}^N f(X_k)=\sum_{i\in E}f(i)\pi(i).$$
\label{thm:ergodic2}
\end{thm}

This theorem assures us that empirical averages converge to probabilistic averages. A proof of this version of Ergodic theorem can be found in \cite{bremaud}. 

\section{Quasi-stationary distributions}
\label{quasi}
In the previous section we have proven that there exists a unique stationary distribution. Note, however, that we took for granted {\sl ergodicity}, which is quite strong assumption. In Sec. \ref{models} and \ref{abm} we describe so-called SIS epidemic spreading models. One of their intrinsic features is the presence of {\sl absorbing state}, i.e. a state from which no states are accesible. We are speaking about {\sl healthy state} -- the situation when are individuals are healthy. Then, there is no way back to any other state of the Markov chain as there is no-one who could possibly infect some other individuals. 

It looks like the whole knowledge presented in the last section is useless for out furher considerations, but happily it is not that bad. The idea is to divide the graph into subgraphs: one is the absorbing state, the second one -- all the other nodes, and analize the latter subgraph separately. Provided a condition, that the system has not yet been driven to the absorbing state, we can analize the dynamics as if our graph was ergodic. In particular, we are able to look for long time behavior and temporary stationary distribution, usually called in the literature {\sl quasi-stationary distribution} or {\sl stationary conditional distribution}.

The formalism of quasi-stationary states is said to be introduced first by A.M.~Yaglom \cite{yag} and analized for Galton-Watson branching process. In a broader case, for discrete-time finite Markov chains, it was analized in seminal paper by J.N.~Darroch and E.~Seneta \cite{qstat}. This is actually the case, which interests us: a short introduction of quasi-stationary distributions presented below is based on this paper. However, there are also generalizations for countable state space Markov chains \cite{qstat-con}, continous-time Markov chains \cite{qstat-cont}, a general state spaces \cite{generalss} along with more involved results for some particulat examples like random walks \cite{randomw}, branching processes \cite{branchp} and more, see \cite{qstat-rev} for a review.

Before we start with the quasi-stationary distribution theory, let us recall the Perron-Frobenius theorem (see e.g. \cite{bremaud,Seneta}), which is used later on.

\begin{df}
A $n \times n$ matrix $A$ is called aperiodic if for any $i$ there exists $k_i\in\mathbb{N}$ such that for all $k\geq k_i$ there holds $(A^k)_{ii}>0$.
\end{df}

\begin{df}
A notion $O(f(x))$ represents a function of $x$ such that there exists $a,b\in\mathbb{R}$, $0<a\leq b<\infty$, and such that for all $x$ sufficiently large the following statement holds: $af(x)\leq O(f(x))\leq bf(x)$.
\end{df}

\begin{thm}[Perron-Frobenius]
 Let $P$ be a $n \times n$ non-negative, irreducible and aperiodic matrix. There exists an eigenvalue $\rho_1$ such that:
 \begin{enumerate}[(a)]
  \item $\rho_1\in\mathbb{R}$ and $\rho_1>0$;
  \item there are positive left and right eigenvectors $u_1$ and $v_1$ associated with $\rho_1$, and chosen such that $u_1v_1^T=1$;
  \item $\rho_1>|\rho_j|$, where $\rho_j$ for $j=2,...,n$ are the other eigenvalues of $P$;
  \item the eigenvector associated with $\rho_1$ is unique up to a constant factor.
  \end{enumerate}
Moreover, let us order the eigenvalues of $P$ in the descending order $\rho_1>|\rho_2|\geq|\rho_3|\geq ... \geq |\rho_n|$. If $|\rho_2|=|\rho_j|$ for some $k\geq 3$, then $m_2\geq m_j$, where $m_j$ is the multiplicity of $\rho_j$, and
$$ P^r=\rho_1^r v_1^T u_1 + O\big( r^{m_2-1}|\rho_2|^r \big).$$
If in addition, $P$ is stochastic or substochastic, then $\rho_1=1$ or $\rho_1\leq 1$ respectively.
\label{thm:pf}
\end{thm}

\vspace{4 mm}
Consider a homogeneous, aperiodic Markov chain $\{X_i\}_{i\geq0}$ with the $(n+1)$-dimensional state space $E=\{a\}\cup T$, where $T$ is irreducible and $a$ is the absorbing state. We enumerate the states as follows: $0$ stands for the absorbing state $a$ and states in $T$ have indices from $1$ to $n$, so that the transition $(n+1)\times(n+1)$ matrix takes the following block structure: 
\begin{equation}
P=\left(\begin{array}{c|c}
1 & \bf{0} \\ \hline
\bf{p}^T & Q
      \end{array}\right).
\label{eq:block}
\end{equation}
Here $1$ is a number, $Q$ is $n\times n$ matrix, $\bf{0}$ and $\bf{p}$ are $n$-component vectors and $\bf{p}$ has at least one entry positive, $\bf{p}\neq 0$. Note, that $\bf{p}$ is a vector of 1-step absorbtion probabilities. Moreover, as $T$-irreducible and $\{X_i\}_{i\geq0}$-aperiodic, $Q$ is substochastic, non-negative, irreducible and aperiodic \cite{Gantmacher}. Let $[q_0(t),\bf{q}(t)]$ be the probability distribution over $n+1$ states at time $t$. We define the distribution conditioned on the event that the system has not yet reached the absorbing state:
\begin{equation}\nonumber
d_j(t)= \mathbb{P}(X_t=j|\tau>t)=\frac{q_j(t)}{1-q_0(t)},
\label{eq:cond-distr}
\end{equation}
where $\tau=\inf\{t\geq0: X_t=0\}$. Note, that $\bf{d}(t)$ is a proper probability distribution on the irreducible set $T$, as $\sum_{j=1}^n d_j(t)=1$.

\begin{df}
A probability distribution $\bm{d}(t)$ on an irreducible set $T\subset E$ is called a quasi-stationary distribution, if $\bm{d}(t+1)=\bm{d}(t)$. We denote it by $\bm{d}$.
\end{df}

\begin{thm}
\label{thm:q-ergodic}
The Markov chain described above has a unique quasi-stationary distribution $\bm{d}$. Moreover, for any initial distribution $\mathbb{\pi}$ on the irreducible set $T$, there exists a large-time limiting distribution and equals $\bm{d}$:
$$ \lim_{t\to\infty} \mathbb{P}\big(X_t=j\big|\tau>t, \mathbb{P}(X_0=i)=\pi_i\textrm{ for }i\in T\big) = d_j. $$
\end{thm}
\begin{proof}
From the general transition rule for the whole Markov chain:
$$ [q_0(t),\bm{q}(t)]P=[q_0(t+1),\textbf{q}(t+1)] $$
and the block-structure (\ref{eq:block}) of the matrix $P$ we conclude that the evolution of $\bm{d}(t)$ is given by $Q$, i.e. $\bm{d}(t)Q=\bm{d}(t+1)$, and so by Theorem \ref{thm:pf} there exist a distribution $\bm{d}$ and a real number $\rho_1>0$ such that:
$$ \bm{d} Q=\rho \bm{d},$$
where $\rho_1$ is the largest eigenvalue of $Q$. It follows also that:
\begin{equation}
 Q^t=\rho_1^t \bm{f}\bm{d}^T+O(t^{m_2-1}|\rho_2|^t),
\label{eq:qn} 
\end{equation}
where $\bm{f},\bm{d}$ are right- and left-handed eigenvectors of $Q$ corresponding to the largest eigenvalue $\rho_2$, see Theorem \ref{thm:pf}. Denote by $\pi_i$ the probability that the process starts in state $i$. The probability that starting from some distribution $\pi_i$ we are in $j$ after time $t$, with the condition that we have not yet reached the absorbing state, is:
\begin{equation}
 \frac{\sum_{i\in T}\pi_i p_{ij}(t) }{\sum_{i\in T}\pi_i (1-p_{i0}(t))},
 \label{eq:cp}
\end{equation}
where we used the notion of $t$-step transition matrix. Using Eq. (\ref{eq:qn}) and the fact that $1-p_{i0}=\sum_{k\in T}p_{ik}$ we rewrite expression (\ref{eq:cp}) as:
\begin{equation}
\frac{\rho_1^n\sum_{i\in T}\pi_i f_i d_j + O(n^{m_2-1}|\rho_2|^n)}
{\sum_{i\in T}\sum_{k\in T}\pi_i f_i d_k +nO(n^{m_2-1}|\rho_2|^n)}\to
d_j+O(n^{m_2-1}\Big(\frac{|\rho_2|}{\rho_1}\Big)^n),
\end{equation}
where we took the limit $n\to\infty$ and used the fact that $\sum_{k\in T}d_k=1$. 
\end{proof}

We obtained the result, that the long-time behavior is given by quasi-stationary distribution $\bm{d}$ and it does not depend on the initial distribution $\pi$. Note, that the relevance of quasi-stationary problem is strongly related to the values of $\rho_1$ and $\rho_2$. If $\rho_1$ is near the unity and $|\rho_2|$ is realatively small, the convrgence of quasi-stationary distribution is fast. We would expect that for $Q$-nearly stochastic, i.e. when vector of 1-step absorbtion probabilities $\bf{p}$ is close to zero. More detailed analysis of quasi-stationary state convergence can be found e.g. in \cite{qstat-rev}. 

By Theorem \ref{thm:q-ergodic} we proved here a kind of quasi-stationary analogy of Theorem \ref{thm:ergodic}: there is an unique distribution, which maybe does not survive forever, but which lasts for a long time. Moreover, there is also a quasi-stationary version of Theorem \ref{thm:ergodic2}, namely that empirical averages converge to probabilistic averages under the condition, that we have not met the absorbing state \cite{quasiergodic}.

For the end of this chapter we define mixing times. In many natural sciences -- particularly in physics or epidemiology -- people deal with some stable stationary states, as they are much easier to cope with than any situation out of equalibrium. Also in epidemic spreading modeling, as in real epidemic problems, we are interested in the final effect of the epidemy. However, it is also pretty interesting, how much does an epidemy need to spread over the system, how much time does it take to reach the (quasi-)stationary state? Or, strictly speaking, to come {\sl close} to it. In order to ask this question in a precise way, we introduce the notion of mixing time.

\begin{df}
For a Markov chain with transition matrix $P$ and stationary distribution $\pi$, the $\epsilon$-mixing time is:
$$T(\epsilon)=\max_{i}\inf\{t:\, \frac{1}{2}\sum_{j=0}^n|p^t_{ij}-\pi_j|\leq\epsilon,\,\forall s\geq t \}. $$
\label{df:mix}
\end{df}

In fact, as we noted in the last section, our aim is often to gain some information about {\sl quasi-stationary distribution}, so we can equivalently put {\sl quasi-stationary distribution} in the previous definition.

\chapter{Epidemic spreading models}
\label{models}
In this chapter we introduce the basic contact network models of epidemic spreading. We identify individuals with nodes of a network and we link them according to the acquaintance structure, i.e. if A knows B, we put a link between them. The graph that represents the network is assumed to be connected. We can also attribute weights to the edges according to the intensity of contacts between two particular individuals. The epidemy spreads over the links of the network, so the construction of the network varies depending on the disease. For instance in order to acquire infection of influenza or HIV two people have to breath the same air or have sex, respectively, so networks for models for spreading of these two disease should differ. This review is based mainly on the approach presented in \cite{corr,DynProc,Networks}. 

The idea is to divide the population into three classes: susceptible (denoted by S, the group of individuals that can be contaminated), infectious (I, individuals that can contaminate others), and recovered (R, who are healthy and resistant). This basic framework can be used to compose models of epidemic spreading. We describe 3 of them: SI, SIS, and SIR, where e.g. SIS stands for a model where a particular individual is susceptible at the beginning, then can be infected, and then can be susceptible again. In general one can consider also plenty of different models, for instance $SI_1I_2$ where after being infected an individual is first strongly contagious (contaminates others with high probability), and after some time -- less contagious. Various kinds of models can be composed according to the particular disease under investigation. This is of course a huge simplification of what actually happens in nature and many biological details are sweeped under the rug, but the main features are captured. And well, at the end of the day this is how physicists and mathematicians do create models.

\section{Homogeneous networks}
\label{homo}
As an initial approach, we impose an assumption that within each class (S, I and R) individuals are identical and homogeneously mixed. This is what we call {\sl the homogeneous assumption}. Let us start with the simplest epidemiological model: susceptible-infected (SI) model. The only possible evolution for each individual is to be initially susceptible (except of one node, ''Patient Zero'') and then infected. There is no way out of infected state so if the contact graph is connected, sooner or later all the individuals are going to be infected. 

Let us introduce the parameters and rules that govern the model. First of them is the probability $\beta dt$ that a susceptible node acquires the infection from an infected neighbor during a time interval $dt$. Parameter $\beta$ is called {\sl spreading rate} and is pathogen dependent. The total probability of getting infected during the time interval $dt$ is $1-(1-\beta dt)^m$, where $m$ is the number of infected neighbors of a given suspectible node. On average, the number of infected neighbors is $m=ki$, where $k$ is the overall number of neighbors of a given node and $i(t)=\frac{I(t)}{N}$ stands for density of infected individuals ($I(t)$ is the number of infected individuals and $N$ is the size of the system). For $\beta dt$ small enough, we apply linearization and thus the term $1-(1-\beta dt)^{ki}$ simplifies to $\beta ki dt$, yielding the \emph{per capita} acquisition rate $\beta ki$. The number of newly contaminated nodes has to depend on the number of nodes that actually can be infected. Now comes our assumption: because of homogeneous mixing the change of density of infected nodes is proportional to the density of susceptible nodes $s(t)=1-i(t)$, where $s(t)=\frac{S(t)}{N}$ ($S(t)$ is the number of susceptible individuals). On the top of that we put average node degree $\langle k \rangle$ in place of individual degrees. Note, that as we are speaking about random process, the values of $I(t)$, $S(t)$ or $R(t)$ are random variables ($R(t)$ is the number of recovered individuals, to be used in the SIR model). Indeed, they usually vary between different runs of the process. Thus it is useful to treat for instance $I(t)$ as a continous variable representing the {\sl expected} number of infected individuals. We can think of this expected value as an average over many runs of the process taken under identical conditions. Finally, the equation describing the dynamics of the model reads as follows:
\begin{equation}
\frac{di(t)}{dt}=\beta \langle k \rangle i(t)\big(1-i(t)\big),
\label{eq:si} 
\end{equation}
\noindent which is of the form of the well-known logistic equation. From the simple SI model, we can proceed to SIS and SIR equations just by adding some extra terms. Let us now consider SIS model (susceptible-infected-susceptible). Now the infected nodes have possibility to recover and be susceptible again. During the time interval $dt$ it happens with probability $\mu dt$, where $\mu$ is called the recovery rate. This assumption yields an extra term in the equation describing the evolution of the SIS model:

\begin{equation}
\frac{di(t)}{dt}=\beta \langle k \rangle i(t)\big(1-i(t)\big)-\mu i(t),
\label{eq:sis} 
\end{equation}
where we still keep fixed the condition $s(t)=1-i(t)$.

In the SIR model (susceptible-infected-recovered) we have three fractions of nodes, so we need more than one equation in order to describe the system. In this variation of epidemic spreading model an individual that has been recovered stays in this state: there is no possibility to transfer to susceptible or infected compartment again. Such a node in fact blocks the way of spreading of epidemy: we may interpret here a recovered node as a resistant or dead one. The equations that describe such a dynamics read as follows:

\begin{eqnarray}
\frac{di(t)}{dt}&=&\beta \langle k \rangle i(t)\big(1-r(t)-i(t)\big)-\mu i(t), 
\label{eq:sir1} \\
%\frac{ds(t)}{dt}&=&-\beta \langle k \rangle i(t)\big(1-r(t)-i(t)\big), 
%\label{eq:sir2} \\
\frac{dr(t)}{dt}&=&\mu i(t), 
\label{eq:sir3} 
\end{eqnarray}
where $r(t)=\frac{R(t)}{N}$ stands for the density of recovered nodes. Note, that Eq. (\ref{eq:sir1}) differs from Eq. (\ref{eq:sis}) only by term $-r(t)$, as in SIR model the density of susceptible nodes is $s(t)=1-r(t)-i(t)$, not $s(t)=1-i(t)$ as for SIS and SI. The final result of SIR procedure is easy to guess: in the end all the individuals are going to be recovered (at least on a connected graph) or still susceptible -- but not infected.

Note that in both SIS and SIR model we have two parameters, $\beta$ and $\mu$, that are in general independent. Let us focus on the relation between them. One can distinguish two extreme cases. When the recovery rate $\mu$ is much bigger than the spreading rate $\beta$, $\mu\gg\beta$, the process is dominated by recovery and driven rapidly into the healthy state. For the contrary, if $\mu\ll\beta$ we can neglect the recovery mechanism and we end up with the SI model with all the nodes infected. It is thus reasonable to ask: where is the transition between these two completely different types of behavior?

In order to answer this question we check what happens at the very beginning of the process, when the fraction of infected nodes is very small: $i(t)\ll 1$. This approximation allows us to neglect all terms higher than linear in $i(t)$, in particular: $i^2(t)\approx 0$. Let us start again with the SI model, Eq. (\ref{eq:si}) transforms to:

\begin{equation}
\frac{d i(t)}{dt}=\beta\langle k \rangle i(t), 
\label{eq:lsi}
\end{equation}
which can be easily solved:

\begin{equation}
 i(t)=i_0 e^{\frac{t}{\tau}},
\label{eq:lsi-sol} 
\end{equation}
where $i_0$ is the initial density of infected nodes and $\tau=(\beta \langle k \rangle)^{-1}$ is the time scale of the SI process. The solution for the full SI equation (\ref{eq:si}) reads:
\begin{equation}
 i(t)=\frac{i_0 e^{\frac{t}{\tau}} }{1+i_0 ( e^{\frac{t}{\tau}}-1)}
\end{equation}
and indeed for short times $t\ll\tau$ we obtain again the result computed for $i(t)\ll 1$, Eq. (\ref{eq:lsi-sol}), and for long times $t\gg\tau$ saturation $i \to 1$ occurs.

In the SI model we know that in the end all the individuals will be infected anyway. For the SIS, however, the procedure presented above yields an interesting conclusion. The equation describing the SIS model (\ref{eq:sis}) transforms to:
\begin{equation}
\frac{d i(t)}{dt}=\beta\langle k \rangle i(t)-\mu i(t).  
\end{equation}
Its solution looks the same as for the SI model, Eq. (\ref{eq:lsi-sol}), but now the time scale changes:
\begin{equation}
 \tau=(\beta \langle k \rangle-\mu)^{-1}.
 \label{eq:homo_t}
\end{equation}
A striking difference in the time scales for the SI and SIS models is that now $\tau$ can be negative. By the Grobman-Hartman theorem, except of the set of parameters $\beta, \mu$ such that $\beta \langle k \rangle-\mu=0$ (in the case of SIS model), we can deduce the behavior of the system in the vicinity of $i(t)=0$ out of the linearization provided above. The situation $\beta \langle k \rangle-\mu=0$ comes out to be a transition point between two opposite pictures. If $\tau$ is negative the exponent in Eq. (\ref{eq:lsi-sol}) will also be negative and epidemy will die out.
On the other hand, for positive $\tau$, the density of infected nodes $i(t)$ increases. Therefore the sign of the characteristic time $\tau$ tells us about the behavior of our system. The border between extinction and outbreak of the epidemy is defined by $\tau=0$. This condition is called the {\sl epidemic threshold}. The epidemy spreads over a macroscopic number (of the order of size of the system) of individuals only {\sl above} the epidemic threshold, i.e. :
\begin{equation}
\tau^{-1}=\mu (R_0-1)>0,
\label{eq:sis_tr} 
\end{equation}
where $R_0=\beta\langle k \rangle/\mu $ is the \emph{basic reproductive rate}. For the case of the SIR model the analysis works in exactly the same way, as $r(t)$ can be considered of the same order as $i(t)$, so the term with $r(t)$ in the linearization of Eq. (\ref{eq:sir1}) vanishes. It comes out, however, that here we cannot apply Grobman-Hartman theorem, but according to the experience driven from simulations (see e.g. \cite{DynProc}) we compute and use the notion of epidemic threshold to examine the behavior of the epidemy. The same argument we apply in Sec. \ref{hetero} and \ref{correlated}.

Note also, that epidemic threshold and time scale actually capture the main characteristic of a network: average degree $\langle k \rangle$, in contrast to classical approach to epidemic spreading modeling described at the beginning of this chapter. In all three models under consideration -- SI, SIS, SIR -- the characteristic time $\tau$ decreases with average degree. It means that the denser a network is, the faster an epidemy spread, which sounds reasonable. Moreover, the basic reproductive rate $R_0$ in SIS and SIR models is linear with average degree $\langle k \rangle$. One of the conclusions for this result could be that for fixed spreading rate of a pathogen $\beta$ for highly connected networks (big $\langle k \rangle$) we have to increase the recovery rate $\mu$ in order to stop the outbreak (lower $R_0$ to be smaller than 1).

It is important to stress, that the condition (\ref{eq:sis_tr}) only makes the outbreak possible, but does not guarantee that it will happen. In other words: the probability of epidemic outbreak is in general less than 1, although we are above the epidemic threshold. Due to stochastic fluctuations in finite system the epidemy is able to die out, which is probable especially at the very beginning of the process, when the probability that the first few infected nodes recover at once is significantly higher than zero. It has actually been estimated, that the extinction probability above the threshold of an epidemy starting with $n$ infected nodes is equal to $R_0^{-n}$ \cite{bailey}.

\section{Heterogeneous uncorrelated networks}
\label{hetero}
The homogenenity assumption has already given us some insight into epidemic spreading dependence on the topology of the network, namely dependence on the characteristic time, and epidemic threshold dependence on average degree $\langle k \rangle$. However, most of the real networks are highly heterogeneous. In the previous section we assumed that all the nodes has the same number of neighbors $\langle k \rangle$, whereas for instance in social networks there are individuals with much more neighbors than an average one \cite{hy} (so-called ''super-spreaders'' in epidemiological terminology, or hubs in general network context). The first step towards this problem can be addressed by \emph{degree block approximation} \cite{DynProc}, i.e. by considering all the nodes with the same degree as statistically equivalent. This, however, does not have to always be the case in reality. Indeed, in particular for some regular lattices or structured networks \cite{klemm, moreno}. Moreove, in the present section we focus on uncorrelated networks only (correlated once are under investigation in the next section). Let us introduce new density variables:
\begin{equation}
i_k=\frac{I_k}{N_k},\,\, s_k=\frac{S_k}{N_k},\,\, r_k=\frac{R_k}{N_k},
\end{equation}
where $N_k$ is the number of nodes with degree $k$ and $I_k,\,S_k,\,R_k$ are the numbers of infected, susceptible and recovered individuals respectively. The total densities of all these three classes of nodes are given by:
\begin{equation}
i=\sum_{k}\mathbb{P}(k)i_k,\, s=\sum_{k}\mathbb{P}(k)s_k,\,r=\sum_{k}\mathbb{P}(k)r_k,
\end{equation}
where $\mathbb{P}(k)=\frac{N_k}{N}$ stands for the degree distribution of a network.

Within this framework let us start again with the simplest SI model. Having Eq. (\ref{eq:si}) in mind we create an equation for the heterogeneous case. The change in time of the fraction of infected nodes with degree $k$, $\frac{di_k(t)}{dt}$, is proportional again to the spreading rate $\beta$, the number of neighbors of a node $k$, the probability that this neighbor is susceptible $1-i_k$ and -- this is crucial -- the density $\Theta_k(t)$ of infected neighbors of vertices of degree $k$. The last term, in homogeneous assumption simply taken as $i(t)$, is for now an unknown function, which we will compute in a while. Equation, that describes the dynamics in the heterogeneous assumption is:
\begin{equation}
 \frac{di_k(t)}{dt}=\beta k \big(1-i_k(t)\big)\Theta_k(t).
\label{eq:si_het} 
\end{equation}

Now we have to face the problem of finding $\Theta_k$. The simplest case of a network we can take is the one with no degree correlation. By {\sl no degree correlation} we mean that the probability of an edge to point {\sl from} a vertex of degree $k$ {\sl to} a vertex of degree $k'$, which we denote by $\mathbb{P}(k'|k)$, does not depend on $k$. Let us recall the exact formula for it, Eq. (\ref{eq:uncor}):
\begin{equation}
\mathbb{P}(k'|k)=\frac{k'\mathbb{P}(k')}{\langle k \rangle }. 
\end{equation}
Employing this expression, we easily find that the form of the density $\Theta_k$ of infected neighbors of vertices of degree $k$ is $\Theta_k(t)=\Theta(t)=\frac{1}{\langle k \rangle }\sum_{k'}k'\mathbb{P}(k')i_{k'}(t)$ \cite{PSV}. But actually most of the authors nowadays use a heuristic modification of this equation (See e.g. \cite{DynProc,corr}), namely:
\begin{equation}
 \Theta_k(t)=\Theta(t)=\frac{\sum_{k'}(k'-1)\mathbb{P}(k')i_{k'}(t)}{\langle k \rangle },
 \label{eq:rem}
\end{equation}
where we have term $(k'-1)$ instead of $k'$ because each infected node uses at least one link to connect with another infected node which has contaminated that first one and thus cannot be used to link to a healthy $k$-degree susceptible node. Note, however, that this reasoning does not work for SIS model, because of the possibility of being susceptible again. 
Note also, that we end up with $\Theta_k(t)$ that, thanks to no degree correlation assumption, does not depend on node degree $k$. Differentiating the last equation and employing there Eq. (\ref{eq:si_het}) we get the evolution equation for $\Theta(t)$`:
\begin{equation}
\frac{d\Theta(t)}{dt}=\frac{\beta}{\langle k \rangle }\Theta(t)\Big(\langle k^2 \rangle -\langle k \rangle 
-\sum_k k^2\mathbb{P}(k)i_k(t)+ \sum_k k\mathbb{P}(k)i_k(t) \Big).
%\frac{d\Theta(t)}{dt}=\frac{\beta}{\langle k \rangle }\Theta(t)\Big(\langle k^2 \rangle -\sum_k k^2\mathbb{P}(k)i_k(t)\Big).
\end{equation}
In the small $i_k$ approximation, $i_k(t)\ll 1$ (as we did for homogeneous case), we obtain:
\begin{eqnarray}
  \frac{di(t)}{dt}&=& \beta k \Theta(t), \\
  \frac{d\Theta(t)}{dt}&=&\beta \Big(\frac{\langle k^2 \rangle}{\langle k \rangle} -1 \Big)\Theta(t).
  %\frac{d\Theta(t)}{dt}&=&\beta \frac{\langle k^2 \rangle}{\langle k \rangle}\Theta(t).
\end{eqnarray}
This set of equations solved for uniform initial condition, $i_k(t=0)=i_0$ for all $k$, gives:
\begin{eqnarray}
\label{eq:si_het_sh}
i_k(t)&=&i_0 \Big(1+ \frac{k(\langle k \rangle-1)}{\langle k^2 \rangle-\langle k \rangle}(e^{t/\tau}-1) \Big), \\
\label{eq:si_het_t}
\tau^{-1}&=&\beta(\kappa-1),
\end{eqnarray}
where $\kappa=\frac{\langle k^2 \rangle}{\langle k \rangle}$ is the heterogeneity ratio as defined in Eq. \ref{eq:kappa}.  

Here comes a prominent impact of the structure of a network on the dynamics of epidemic spreading. Note that for homogeneous networks with a Poisson degree distribution we have $Var(k)=\langle k \rangle$ and so $\kappa=\langle k \rangle+1$ and we recover the result for homogeneous case, Eq. \ref{eq:homo_t}. On the other hand, for heterogeneous networks, $\kappa$ is very large, for instance: in power law distribution scale-free networks we have shown that in the real-world cases, when the exponent vary between 2 and 3, $\langle k^2 \rangle$ goes to infinity, and so does $\kappa$. In the SI model case it just means that epidemic spreading runs horribly fast, but in the SIR case the result will be much more significant, and, so to say, the conclusion even scary.

Because of a remark to Eq. \ref{eq:rem} the SIS case is little bit different, but for SIR we go exactly the same way as for SI, just puting the extra recovery term $-\mu i_k(t)$ into the evolution equation of $i_k(t)$ and making the same approximation as for SI. It has been found \cite{DynProc,jtb} that the time scale of the process behaves like:
\begin{equation}
 \tau^{-1}\sim\beta\kappa-(\mu+\beta).
 \label{eq:sir_het_t}
\end{equation}

Under the homogeneous assumption we were looking for epidemic threshold: we checked in which conditions the characteristic time changes its sign, see Eq. \ref{eq:sis_tr}. We can do the same in this case: we expect an epidemic outbreak for time $\tau$ greater than zero, which translates to:
\begin{equation}
 \frac{\beta}{\mu}\geq\frac{1}{\kappa-1}.
\end{equation}
This yields some finite threshold value for homogeneous networks. However, for heterogeneous networks the situation is strikingly different. Let us focus on scale-free networks described in Sec. \ref{sf}, i.e. power law degree distribution graphs with the exponent $2\leq \alpha \leq 3$. As we have already discussed, for this kind of networks the ratio $\kappa$ diverges, so the right hand side of the equation above is zero. Therefore coefficients $\beta$ and $\mu$ (by definition: positive) can take any finite value! Note, that as we have already mentioned in Sec. \ref{sf}, those networks appear to be quite good model of our real social networks. The conclusion is that our social links build up a perfect environment for epidemies to spread. Of course, again: in reality networks are finite, so $\kappa$ is finite as well. However, if our real social networks really are that heterogeneous as power-law graphs are (what seems to be true: for instance, if we speak about sexually transmitted diseases, the web of 
sexual 
contacts actually is scale-free \cite{sex}), then $\kappa$ can be huge anyway and then the theoretical result of epidemic threshold would not be zero, 
but {\sl almost} zero, which is quite dangerous as well.

So is it actually that bad? The hubs in social networks help a lot the epidemy to spread fast through the whole system. But it is a double-edged sword: from the other hand the presence of hubs in networks allows us to provide efficient strategies of vaccination. This topic goes far beyond our elaboration, but interested reader can see examples of those strategies of targeted immunization in \cite{im2,im3,im1}.

%dodać przy informacji o braku punktu krytycznego dla bezskalowych sieci, że to nie jest cecha tylko i wyłącznie bezskalowości, ale samopodobieństwa: Serrano  et al. PRL 100, 078701, 2008 - tu tego nie ma wprost...

As a supplement to this section, let us go back to the case of SIS model, which we abandoned at Eq. \ref{eq:rem}. The thing that differs the SIS from others is the form of $\Theta$ function. In the SIS model nodes have possibility to be susceptible again, therefore there is no extra $(-1)$ term in $\Theta$ function, i.e.:
\begin{equation}
 \Theta_k(t)=\Theta(t)=\frac{\sum_{k'}k'\mathbb{P}(k')i_{k'}(t)}{\langle k \rangle },
 \label{eq:rem2}
\end{equation}
which consequently yields the epidemic outbreak condition given by \cite{DynProc,jtb}:
\begin{equation}
 \frac{\beta}{\mu}\geq\frac{1}{\kappa}=\frac{\langle k \rangle}{\langle k^2 \rangle}.
\end{equation}
In fact the result is more or less the same and the conclusion about weakness of scale-free network in case of epidemic attack still holds.

\section{Correlated networks}
\label{correlated}
In the last section we relaxed the assumption of homogeneity, but kept no-correlation. But many real networks manifest non-trivial degree correlations, so the next reasonable step we can do is to consider correlated networks. We leave, however, one of simplifications: as in the last section we consider all nodes with the same degree as statistically equivalent.

As we have already seen the difference of dealing with the SIS and, in contrast, SI and SIR model, let us from now focus on the SIS model only, following \cite{corr}. The evolution equations are similar to the no-correlation case: the only difference is that now function $\Theta_k(t)$ really depends on node degree $k$:
\begin{eqnarray}
 \frac{di_k(t)}{dt}&=&\beta k \big(1-i_k(t)\big)\Theta_k(t)-\mu i_k(t), \\
 \Theta_k(t)&=&\sum_{k'}i_{k'} \mathbb{P}(k'|k).
\end{eqnarray}
Therefore, the probability that a link departing from a node of degree $k$ is pointing to an infected node $\Theta_k(t)$
is a sum over probabilities of pointing to a node with degree $k'$ times the probability that this node is infected, $i_{k'}(t)$. Putting the latter equation into the previous one gives the equation fully describing the evolution of the SIS process:
\begin{equation}
 \frac{di_k(t)}{dt}=\beta k \big(1-i_k(t)\big) \sum_{k'}i_{k'} \mathbb{P}(k'|k) -\mu i_k(t).
\end{equation}
Following our regular procedure, we approximate the equation for small densities $i_k$ and see what kind of conclusions we can draw out of this linearization:
\begin{eqnarray}
\label{eq:system}
  \frac{di_k(t)}{dt}&=&\sum_{k'} L_{kk'}i_{k'}(t),\\
  L_{kk'}&=&-\mu \delta_{kk'}+\beta k \mathbb{P}(k'|k),
\end{eqnarray}
where $\delta_{kk'}$ stands for the Kronecker delta. One of the solutions of this system of first order differential equations is obviously the healty state, $i_{k}=0$ for all $k$. Let us however look for another one. We define {\sl connectivity matrix} $C_{kk'}=k \mathbb{P}(k'|k)$, such that now $L_{kk'}=-\mu \delta_{kk'}+\beta C_{kk'}$. This matrix has only real eigenvalues. In order to show that, we denote $P_k=P(k)$ to give an intuition of matrix-vector notation and define a rescaled connectivity matrix:
$$\tilde{C}_{kk'}=\left\{ \begin{array}{ll}
                   \frac{C_{kk'}P_{k'}}{\sqrt{P_{k}P_{k'}}} & \textrm{for } P_{k},P_{k'} \neq 0\\
		   0 & \textrm{otherwise} 
                 \end{array} \right. 
$$ 
Due to the detailed balance condition, Eq. (\ref{eq:db}), the rescaled matrix $\tilde{C}_{kk'}$ is symmetric, and thus it has only real eigenvalues. Note, however, that matrices $C$ and $\tilde{C}$ have common set of eigenvalues. Take an arbitrary eigenvector of $C$, call it $v_k$: $\sum_{k'}C_{kk'}v_{k'}=\Lambda v_k$. Then, adequatly rescaled vector:
$$\tilde{u}_{k}=\left\{ \begin{array}{ll}
                   \frac{v_{k}}{ \sqrt{P_{k}} } & \textrm{for } P_{k}\neq 0\\
		   0 & \textrm{otherwise} 
                 \end{array} \right. 
$$ 
is an eigenvector of $\tilde{C}$: $\sum_{k'}\tilde{C}_{kk'}u_{k'}=\sum_{k'} \frac{C_{kk'}P_{k'}}{\sqrt{P_{k}P_{k'}}} \frac{v_{k'}}{ \sqrt{P_{k'}} }=\Lambda  \frac{v_{k}}{ \sqrt{P_{k}}}=\Lambda u_k$ where we assumed all $P_{k}\neq 0$ without loss of generality.

The result of differential equation (\ref{eq:system}) is given by a linear combination of exponential functions of the form $exp((-\mu+\beta\Lambda_i) t)$, where $\Lambda_i$ are eigenvalues of connectivity matrix $C$, see e.g. \cite{malham}. The dominant behavior of the density of infected nodes will be thus given by the largest eigenvalue $\Lambda_m$, $i(t)\sim e^{(-\mu+\beta\Lambda_m) t}$. Keeping this in mind we follow again the procedure from the previous sections. Characteristic time of the SIS process for network with correlation will be proportional to $1/((-\mu+\beta\Lambda_m))$ and therefore the epidemic outbreak condition read:
\begin{equation}
 \frac{\beta}{\mu}\geq\frac{1}{\Lambda_m}.
\end{equation}
Here again, as for the case of uncorrelated networks, it can be shown \cite{corr} that for scale-free networks the largest eigenvalue $\Lambda_m$ can be lowerbounded by infinity, and thus we end up with {\sl null epidemic threshold}, which implies the same scary conslusion as we saw in the last section.

It has to be stressed, that this interesting result involved not only heterogeneity of the system, but also correlations. However, as we pointed out at the beginning of the previous section, we still keep the assumption that the nodes with the same degree $k$ are statistically equivalent. There is also a way to relax this assumption. In case of homogeneous networks we used only the most basic information about the network: the average degree $\langle k \rangle$. Then we made a step forward to heretogeneous networks, where in fact we used in addition the second moment of degree distribution: $\langle k^2 \rangle$. Further on, in order to include correlation into our analysis, we took matrix of conditional probabilities $\mathbb{P}(k'|k)$. Finally, we can also use {\sl the whole} information about the network, namely: explicit form of adjacency matrix $A_{ij}$. Then the evolution equation of the SIS model reads:
\begin{equation}
 \frac{di_j(t)}{dt}=\beta \big(1-i_j(t)\big) \sum_{l}A_{jl}i_{l} -\mu i_j(t),
\end{equation}
where $i_j$ stands for expected value of node $j$ to be infected. The adjacency matrix can be unweighted (only $0,1$ entries) or weighted, with weights that are associated with the frequency of conntacts between nodes. Note, however, that this is still not the full formulation, as correlations between infected and susceptible nodes are neglected (in order to include them we should use a kind of $\Theta$-function known from the previous sections). Analysis of this approximation is technically much more involved, however the general idea is similar to the case of connectivity matrix $C$: to draw some conclusions out of the largest eigenvalue of a matrix, this time -- adjacency matrix $A_{ij}$. The result is actually quite interesting, as it shows that disease, in a sense, localizes on a finite number of vertices. Interested reader is advised to see \cite{Localization} for more details.

\chapter{Single infection epidemic spreading model}
\label{abm}
In this chapter we introduce single infection epidemic spreading model (SIESM), proposed by the author of this elaboration \cite{Ganczarek}. SIESM exhibit some significant features that qualitatively differ from the models described in Chapter \ref{models}. In order to notice those differences easily, we begin with a short overview of what we have learned so far. Then, in Sec. \ref{des} we describe proposed model of epidemic spreading with at most one infection per time step. In Sec. \ref{ana} we provide theoretical analysis of the model: epidemic threshold, quasi-stationary state and mixing time. Simulations are presented in Sec. \ref{sim}, while Sec. \ref{conclu} is dedicated to some final conclusions.

\section{Introduction}
The most appropiate models for epidemic spreading are those based on dynamical processes on particular graphs rather than those defined by fenomenological differential equations \cite{may,NewmanSIAM}. Within this approach the nodes of a network are usually considered as individuals, who are connected with each other by vertices corresponding to social links.  Although some authors use continous time simulations (see e.g. \cite{Boguna}), the approach presented commonly (see \cite{DynProc} for a review) is based on the idea that at each discrete time step a particular node of the network can contaminate each of its neighbors with some finite probability $p$. The whole set of vertices is being divided into compartments, usually referred to susceptible (S), infected (I) and recovered (R) individuals. Different models can be built out of those $S,I,R$ letters,  but the general mechanism stays more or less unchanged, as we have described in Sec. \ref{models}. In the most basic approach people assume individuals to 
be identical and homogeneously mixed (homogeneous 
assumption, Sec. \ref{homo}). In order to take into account heterogeneity of the system a kind of block approximation has been used \cite{PSV}, treating nodes with the same degree as statistically equivalent (see Sec. \ref{hetero}). This is not always enough, as some real networks manifest degree correlation, namely: the conditional probability, that two vertices of degree $k,\,k'$ are connected depends on both degrees $k,\,k'$ \cite{cor1}. The next step thus is to take into account correlation \cite{BPSV}. Finally, one can employ whole adjacency matrix describing the graph we analyze \cite{FullM,cor2,cor3} (Sec. \ref{correlated}). The validity of all these approaches is still under investigation, see e.g. \cite{Localization}. Note, however, that all those variations listed above work on equations describing relationships between probability vectors. In particular, for the last example, the system is being described by $i_j$ -- probability, that $j$-th node is infected. The problem is that there is not a 
single moment when a particular vertex is -- say -- 0.41 infected. A vertex can be either infected (1) or not (0). This problem has been already noticed by Petermann and De Los Rios \cite{Paulo}. On the top of that, it has also been realized (see e.g. \cite{finite}) that finite sizes of real network can have a strong impact onto epidemic threshold results, which are usually being computed in the large-n limit (so-called thermodynamical limit).

Here we introduce another model for epidemic spreading and analize it with completely different approach. Let us focus on sexual transmitted diseases. For this case the assumption that a particular node is able to contaminate more than one of its neighbors during a time step seems not to be the most suitable one. Bearing this idea in mind we develop a single infection epidemic spreading model. It is designed intentionally for finite networks only, no thermodynamical limit is regarded here, but the effect of the network size is under investigation instead.

\section{Model description}
\label{des}
Consider a connected, unweighted graph $G$ with $n$ vertices enumerated by $i=1,\,\dots,\,n$, and encoded by a stochastic transition matrix $P$. The model will be of the SIS kind: all the individuals are at the beginning considered as susceptible (S). After contamination they become infected (I) but they still have a chance to recover and be susceptible again.

We start with the all but one nodes susceptible. The one which is infected is chosen at random. At each time step we choose randomly, with identical probability $\frac{1}{n}$, a node $i$. Then we choose its neighbor according to the transition matrix $P$, i.e. there is $p_{ij}$ chance that we take vertex $j$. If one of these two individuals $i$ or $j$ is infected, it contaminates the second one with probability $z$. At the end of each time step we allow each infected node -- except of the one that has been contaminated during the present time step -- to recover with probability $r$. 

Each node of the contact graph $G$ can take any of 2 possible states (susceptible or infected), hence the state of the whole system at time $t$ is described by the set of infected nodes $I(t)$. Therefore, the cardinality of the state space equals $2^n$. It is visible, that in order to drive the system from a time step $t$ to $t+1$ we only need to know in what state it is at time $t$. We express it by the following theorem:

\begin{thm}
The stochastic process of the SIESM model is a homogeneous Markov chain. 
\end{thm}

Note, that this method restricts not only each infected node to contaminate at most one of its neighbor. In fact we restrict the whole dynamics to at most one contamination per time step. In this way we would like to examine the impact of the number of individuals in -- so to say -- ``overcrowded'' systems. One can imagine a real-world situation, where the access to some devices or services is strictly limited by a finite value, while the number of people is increasing (e.g. only one doctor in a small, but growing town). Here we focus on the following example. In large, academic cities there are often big flats situated in old tenement houses, inhabited by rather big amounts of students, who live with 3-4 roommates per chamber. As there is no space for privacy in this way of living, they sometimes devote one room in the flat to be a so-called {\sl sexroom}, so contamination by sexually transmitted diseases can take place at most once per time step (say: per night). It seems to be a good example of a system 
which can be described by our model. Note, that in a system such as this group of students in a tenement house, the role of the network size is crucial. 

\section{Model analysis}
\label{ana}
%Let $\{Y_t\}_{t\geq0}$ be a discrete-time Markov chain with a finite state space, describing the the SIESM model. 
We enumerate the state space in such a way, that the absorbing state, which is the state when all the nodes are healthy, has index $0$. In order to mathematically describe the model we define $X_j(t)$ which takes the value $1$ if the node $j$ is being contaminated by one of its neighbors at the time step $t$, and $0$ otherwise. We will be interested in the expected value of $X_j(t)$ with a condition that the set of infected nodes consists of some particular vertices. It can be thought of as the effective probability of contamination during a time step. Note, that in contrast to contamination probability $z$, expected value of $X_j(t)$ varies with time.

There are two independent ways of contaminating $j$-th node during one time step. Either we choose $j$-th node (with probability $\frac{1}{n}$) and then one of its infected neighbor (with probability $\sum_{k\in I(t)}p_{jk}$) or we choose $j$-th node's neighbor (with probability $\frac{1}{n}$ for each one) and then we pick $j$-th node (it happens with probability $p_{kj}$ for a particular node $k$, so $\sum_{k\in I(t)}p_{kj}$ for all of them). Summing up we obtain:
\begin{equation}
  \mathbb{E}(X_j|I(t),\tau>t)=\frac{z}{n}\Big(\sum\limits_{k\in I(t)}p_{jk}+\sum\limits_{k\in I(t)}p_{kj}\Big),
\label{eq:exp}
\end{equation}
where both terms are multiplied by contamination probability $z$ and $\tau$ is defined as $\tau=\inf\{t\geq0: |I(t)|=0\}$, same as in Eq. \ref{eq:cond-distr}. Note the way we use the condition that $\tau>t$: as we take for granted that we have not met yet the absorbing state, the two ``mechanisms'' of contamination actually work. Otherwise, if $t$ was greater or equal $\tau$, we would put $0$ on the right hand side of Eq. (\ref{eq:exp}).

We are, however, interested in the behavior of the whole system, not one node only. Let us define $D\big(t\big|I(t)\big)=\mathbb{E}\Big(|I(t+1)|-|I(t)|\,\Big|\,I(t),\tau>t\Big)$. Due to additivity of expected value we can write:
\begin{equation}
D\big(t\big|I(t)\big)=\sum\limits_{j\notin I(t)}\mathbb{E}(X_j|I(t),\tau>t)-r|I(t)|,
\end{equation}
where, apart from adding all the $\mathbb{E}(X_j|I(t),\tau>t)$ terms, we substract the term responsible for healing: number of infected nodes multiplied by recovery probability $r$. Moreover, we define $D(t)=\mathbb{E}\Big(|I(t+1)|-|I(t)|\,\Big|\tau>t\Big)$ -- expected value of change of the number of infected nodes. We have then $D(t)=\sum_{I(t)}D\big(t\big|I(t)\big)\mathbb{P}\big(I(t)\big)$. Using Eq. (\ref{eq:exp}) we immediately conclude:
\begin{equation}
D\big(t\big|I(t)\big)=\frac{z}{n}\Big(\sum\limits_{k\in I(t),\,j\notin I(t)}p_{jk}+\sum\limits_{k\in I(t),\,j\notin I(t)}p_{kj}\Big)-r|I(t)|.
\label{eq:d}  
\end{equation}
The equation above defines the dynamics of the process. Unfortunately, sums $\sum_{k\in I(t),\,j\notin I(t)}p_{jk}$, $\sum_{k\in I(t),\,j\notin I(t)}p_{kj}$ are not known in general, as they strongly depend on the shape of the set $I(t)$. We will show, however, that we are able to derive exact result for epidemic threshold for any graph and quasi-stationary state for some special cases.

\subsection{Epidemic threshold}
Our first aim is to find out the epidemic threshold for the process described above. We are interested in a relation of model parameters $n,\,z,\,r$ that defines a border between two situations: dropping and rising of the number of infected nodes in the beginning of the process. However, as we noted in Sec. \ref{models}, epidemic threshold is defined in the limit of large size of a network. Here, as we focus on finite networks only and the impact of their size onto the process, we need a different definition of epidemic threshold. In order to construct it, we  follow the method of computing thresholds for infinite networks. In particular, the Eq. (\ref{eq:sis_tr}) we obtained using assumption of small number of infected nodes, which is true in the vicinity of $t=0$. The smallest finite number of infected nodes we can take is 1, so we define epidemic threshold as follows:
\begin{df}
Epidemic threshold for the finite SIESM model is the relation between parameters $z,r$ and network size $n$ such that $D(0)=0$. 
\end{df}

\begin{thm}
Epidemic thresholf for the finite SIESM model is given by:
$$\frac{z}{r}=\frac{n}{2}.$$
\label{thm:trsh}
\end{thm}
\begin{proof}
By definition of the model $|I(0)|=1$. Denote the only infected node by $l$, we have $\mathbb{P}(I(0)=\{l\})=\frac{1}{n}$. Then also $\sum_{k\in I(0),\,j\notin I(0)}p_{kj}=\sum_{j\neq l}p_{lj}=1$ (as $P$ -- stochastic) and $\sum_{k\in I(0),\,j\notin I(0)}p_{jk}=\sum_{j\neq l}p_{jl}$, hence:
$$ D(0|\{l\})=\frac{z}{n}\Big(1+ \sum_{j\neq l}p_{jl}\Big) -r.$$
Futhermore, using the fact that for all $j$ there is $p_{jj}=0$:
$$ D(0)=\frac{1}{n}\sum_{l=1}^n\frac{z}{n}\Big(1+ \sum_{j=1}^n p_{jl}\Big) -r=\frac{z}{r}\Big(1+ \frac{1}{n}\sum_{j=1}^n\sum_{l=1}^n p_{jl}\Big)-r=2\frac{z}{n}-r.$$
The condition $D(0)=0$ indicating the epidemic outbreak completes the proof.
\end{proof}

\subsection{Quasi-stationary distribution}
\label{amb-qs}
Let us now turn to the quasi-stationary distribution problem. In our model there exists a unique absorbing state: the state, when all the nodes are healthy -- there is no way out of this situation, as there is no one who could contaminate others. Moreover, the set of all the states of this Markov chain without the absorbing state, is irreducible. Indeed, as we assume the network of acquaintances to be connected, the epidemy is able to reach all the vertices. There is thus a finite probability to achieve in some time the state, when all the vertices are infected. From this state the system is able to go directly to any other state, because in each time step we recover each individual \emph{independently} with probability $r$. It is thus possible to obtain any configuration of the infected nodes set $I(t)$.  Therefore, by Theorem \ref{thm:q-ergodic}, above the epidemic threshold we anticipate our system to stay at some non-zero state for a long time. 
\begin{thm}
There exists a unique quasi-stationary distribution for the SIESM model.
\end{thm}
The variable of interest in epidemic spreding is the density of infected nodes in this quasi-stationary state. Due to the ergodic theorem for absorbing Markov chains (see discusion in Sec. \ref{quasi}), it actually can be our indicator of achieving the stationarity: the level of infected nodes staying -- on average -- unchanged. In practice we also take into account, that due to statictical fluctations in finite real or simulational systems, the epidemy may die out even above the threshold.    

The quasi-stationary fraction of infected nodes in general case (not specifying any particular shape of the graph) is not as easy to find as the threshold calculated in the previous section. What we basically have to do is to use once again all the formalism presented above and find the solution for the equation $D(t)=0$ without the constraint $|I(t)|=1$. The problem is to compute the sum $\sum_{k\in I(t),\,j\notin I(t)}p_{jk}$, which depends strongly on the shape of the set of infected nodes $I(t)$. We will thus estimate only the quasi-stationary distribution condition for general case. However we provide also exact solutions for special cases of complete graph and uncorrelated homogeneous graph.

In order to perform estimation of the quasi-stationary state, we use the notion of graph conductance. Let $I_s$ be the expected number of infected nodes in the quasi-stationary state. According to the Def. \ref{df:con} we need to analyze separately cases with the quasi-stationary fraction of infected nodes $i_s=\frac{|I_s|}{n}$ smaller and greater than $\frac{1}{2}$. Consider first $i_s\geq\frac{1}{2}$. Then also $|I_s|\geq n-|I_n|$ and, using Eq. (\ref{eq:d}), we lowerbound $D\big(t\big|I(t)\big)$:
\begin{equation}
D\big(t\big|I(t)\big)\geq \frac{2z}{n} \Phi(P)(n-|I(t)|)-r|I(t)|.
\label{eq:gend}
\end{equation}
Bounding the latter expression in Eq. (\ref{eq:gend}) from zero we find that $D\big(t\big|I(t)\big)$ is positive for $\frac{1}{2}\leq i\leq\Big(1+\frac{rn}{2z\Phi(P)}\Big)^{-1}$, therefore the quasi-stationary fraction $i_s$ must be higher than this:
\begin{equation}
i_s\geq\Big(1+\frac{rn}{2z\Phi(P)}\Big)^{-1}.
\end{equation}
Let us now focus on the opposite case, namely $i_s\leq\frac{1}{2}$, $|I_s|\leq n-|I_s|$.
We again lowerbound $D\big(t\big|I(t)\big)$ using Eq. (\ref{eq:d}):
\begin{eqnarray}
D\big(t\big|I(t)\big)\geq \frac{2z}{n} \Phi(P)|I(t)|-r|I(t)|\geq\frac{2z}{n} \Phi(P)|I(t)|-r(n-|I(t)|).
\label{eq:dfi}
\end{eqnarray}
Bounding right hand side of Eq. (\ref{eq:dfi}) from zero, we conclude analogically to the situation above:
\begin{equation}
i_s\leq\Big(1+\frac{2z\Phi(P)}{rn}\Big)^{-1}.
\end{equation}
This result, however mathematically correct, appears to be quite useless: the value of $\Phi(P)$ is usually much lower than the sums that it approximates ($\sum_{k\in I(t),\,j\notin I(t)}p_{jk}$, $\sum_{k\in I(t),\,j\notin I(t)}p_{kj}$) during the process. Let us thus work out exact results for some special cases.

\subsection{Special cases}
\subsubsection{Complete graph}
\label{cgsec}
For complete graphs we are able to find the exact solution of quasi-stationary state problem. Note, that for this special case:
\begin{equation}
\sum\limits_{k\in I(t),\,j\notin I(t)}p_{jk}=\sum\limits_{k\in I(t),\,j\notin I(t)}p_{kj}=\frac{|I(t)|(n-|I(t)|)}{n-1},
\end{equation}
as each of $|I(t)|$ infected nodes is linked to each of $(n-|I(t)|)$ susceptible nodes by an edge chosen with probability $\frac{1}{n-1}$ as each node has $(n-1)$ neighbors. We can thus find explicit and exact condition for $D\big(t\big|I(t)\big)=0$. From Eq. (\ref{eq:d}) we get:
\begin{equation}
i(t)_{s}=1-\frac{r(n-1)}{2z}.  
\label{eq:fg-stat}
\end{equation}

\subsubsection{Uncorrelated homogeneous graph}
\label{largen}
Let us consider now uncorrelated homogenous graph (see Sec. \ref{homo} and \ref{hetero}). Specifically for the present setup, homogeneity means that average number of links between sets of some fixed sizes depends only on these sizes.
\newline\indent
Bearing these assumptions in mind, let us compute expected values of the two sums from Eq. (\ref{eq:d}):
\begin{equation}
\mathbb{E}(\sum_{j\in I(t),\,l\notin I(t)}p_{lj})=\frac{\mathbb{E}(k)}{n-1}\sum_{j\in I(t),\,l\notin I(t)}\mathbb{E}(\frac{1}{k}|k\geq 1) 
 = \frac{\mathbb{E}(k)}{n-1}\mathbb{E}(\frac{1}{k}|k\geq 1)|I(t)|(n-|I(t)|), 
\label{eq:unh}
\end{equation}
where we put $\mathbb{E}(k)/(n-1)$ for the expected value of existence of a link between two vertices. We substract 1 from $n$ as a node cannot be connected with itself. We employ Eq. (\ref{eq:unh}) in the condition $D(t)=0$ and get the quasi-stationary infected nodes density:
\begin{equation}
 i_s=1-\frac{r(n-1)}{2z\langle \frac{1}{k} \rangle\langle k \rangle},
\label{eq:gnp-stat}
\end{equation}
where we denote $\langle k \rangle=\mathbb{E}(k)$ and $\langle 1/k \rangle=\mathbb{E}(1/k)$. Specifically, for $G(n,p)$ random graph (with the well-known binomial degree distribution) the product of $\langle \frac{1}{k} \rangle\langle k \rangle$ goes to 1. In this case the latter result (\ref{eq:gnp-stat}) recovers the solution for complete graphs (\ref{eq:fg-stat}). Moreover, $G(n,p)$ graphs are indeed uncorrelated in the limit of large $n$ \cite{gnp-nocor}, so we expect $G(n,p)$ behaving like complete graphs for large $n$.

\subsection{Mixing time}
\label{ms}
In this chapter we are interested in mixing time for the SIESM model, i.e. the time needed by the process to reach the quasi-stationary distribution. As we have already noted above, we are in fact speaking about the indicator of quasi-stationarity -- the density of infected neighbors. 

In order to gain some intuition about what quasi-stationary state and mixing time are, see Fig. \ref{fig:singlerun}. We can distinguish there two regimes of completely different character: the regime of rapid increase in the number of infected nodes and the regime of stabilization. 

\begin{center}
\begin{figure}[!h]
  \includegraphics[height=6.5cm]{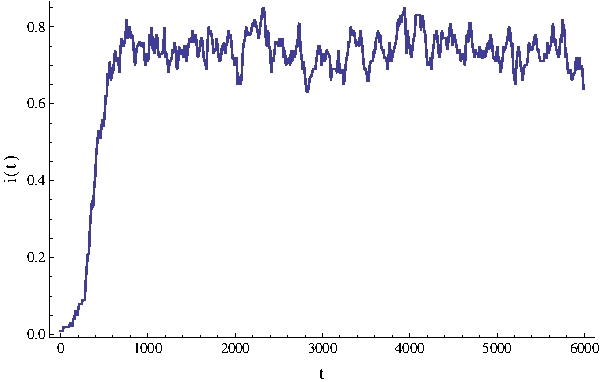}
  \caption{An example of a single run for epidemic spreading defined by the SIESM model, with parameters $z$=1, $r$=0.005, random graph $G(n,p)$ of the size $n$=100 and $p$=0.5. Density of infected nodes increases rapidly in the initial stage, but then, after time $t\approx 600$ (this can be thought of as the mixing time) starts oscillations about some fixed value. }
\label{fig:singlerun}
\end{figure}
\end{center}

Below we prove a general theorem restricting mixing time for any graph. The proof is inspired by related considerations for gossip spreading done by Shah \cite{shah}. The most significant complication which arises here in $SIS$-kind epidemic spreading model is recovery. In the gossip spreading problem we just have subgraph of nodes, that know about the gossip. In case of epidemy with recovering some nodes are being randomly excluded from the subgraph (by recovery), in particular this subgraph can fall apart into smaller compartments.

\begin{thm}
The mixing time $T$ for the SIESM model is logarithmic with the size $n$, i.e.:
$$ T(\epsilon)=O\Big(\log n + \log \epsilon^{-1}\Big). $$
\label{thm:1}
\end{thm}

\begin{proof}
Without loss of generality we assume that $I_s\geq \frac{n}{2}$. We divide the proof into two parts, considering separately two stages of the process evolution: $|I(t)|\leq\frac{n}{2}$ and $|I(t)|\geq\frac{n}{2}$. Let us first consider $|I(t)|\leq\frac{n}{2}$.
We recall first the result stated in Eq. (\ref{eq:dfi}):  
$$D\big(t\big|I(t)\big)\geq \frac{2z}{n} \Phi(P)|I(t)|-r|I(t)|.$$
Denote now by $\Lambda$ the smallest time $t$ such that the number of infected nodes exceeds $\frac{n}{2}$:
$$\Lambda=\inf\{t:|I(t)|>\frac{n}{2}\}.$$
Let us also denote $\Lambda \land t = \min(\Lambda,t).$ Note, that as long as $|I(t)|\leq\frac{n}{2}$, we have $\Lambda \land (t+1)=\Lambda \land t+1$. Recall now the general feature for any smooth convex function $g$ -- for any $x_1,\,x_2\in\mathbb{R}$ we have:
\begin{equation}
 g(x_1)\leq g(x_2)+g'(x_1)(x_1-x_2).
\end{equation}
Let us take: $g(x)=\frac{1}{x}$, $x_1=|I(t+1)|$ and $x_2=|I(t)|$, then:
\begin{equation}
\frac{1}{|I(t+1)|}\leq\frac{1}{|I(t)|}-\frac{1}{|I(t+1)|^2}\Big(|I(t+1)|-|I(t)|\Big).
\label{eq:appl-conv}
\end{equation}
By construction of the process we have:
$$|I(t+1)|\leq|I(t)|+1=d |I(t)|,$$
where $1\leq d\leq 2$. We define as usual $\tau=\inf\{t\geq0: |I(t)|=0\}$ and continue with Eq. (\ref{eq:appl-conv}):
\begin{eqnarray}\label{eq:est}
&\mathbb{E}\Big( \frac{1}{|I(t+1)|}\Big| I(t), \tau>t\Big) \leq  
\mathbb{E}\Big(\frac{1}{|I(t)|} - \frac{1}{d^2|I(t)|^2}\Big(|I(t+1)|-|I(t)|\Big)\Big| I(t), \tau>t\Big)\leq&  \\ \nonumber  
&\frac{1}{|I(t)|}-\frac{1}{d^2|I(t)|^2}\Big( \frac{2z}{n} \Phi(P)|I(t)|-r|I(t)| \Big)\leq 
 \frac{1}{|I(t)|}\Big(1- (\frac{2z}{n} \Phi(P)-r)d^{-2} \Big) \leq  \frac{1}{|I(t)|} \exp(-\frac{1}{d^2}(\frac{2z}{n} \Phi(P)-r)),& 
\end{eqnarray}
where in the second inequality we used Eq. (\ref{eq:dfi}) and the definition of $D(t)$, Eq. (\ref{eq:d}). Note, that also after second inequality we made implicit the condition that $|I(t)|$ is non-zero, due to $\tau>t$. In the last line we used the fact that $1-x\leq \exp(-z)$. Let us now define:
\begin{equation}\label{eq:where}
\zeta(t)=\frac{\exp(at)}{|I(t)|},\,\,
\textrm{where}\,\,\,a=\frac{1}{d^2}(\frac{2z}{n} \Phi(P)-r).
\end{equation}
We show that $\zeta(t)$ is a supermartingale with respect to $\{I(t)\}_{t\geq 0}$ (see Def. \ref{df:sup}). As the only component of $\zeta(t)$ which is a random variable is $I(t)$ and as the process we analyze is Markovian and as $\Lambda \land (t+1)=\Lambda \land t+1$, it is enough to show that $\mathbb{E}(\zeta(\Lambda \land (t+1)) | I(\Lambda \land t), \tau>t )\leq \zeta(\Lambda \land t)$. We do it using Eq. (\ref{eq:est}):
\begin{eqnarray}
&\mathbb{E}\Big(\zeta(\Lambda \land (t+1)) \Big| I(\Lambda \land t), \tau>t \Big)= 
\exp\Big((\Lambda \land t) a\Big) \exp(a)\mathbb{E}\Big(\frac{1}{|I(\Lambda \land t+1)|}\Big|I(\Lambda \land t), \tau>t\Big)\leq&\\ \nonumber
& \leq\exp\Big((\Lambda \land t) a\Big) \exp(a)\frac{1}{|I(\Lambda \land t)|} \exp(-a)= \zeta(\Lambda \land t).&
\label{eq:gnp-pop}
\end{eqnarray}
As $\zeta(t)$ is a supermartingale we conclude that $\mathbb{E}(\zeta(\Lambda \land t))\leq\mathbb{E}(\zeta(\Lambda \land 0))=1$.
Furthermore, as we restrict ourselves to $|I(t)|\leq\frac{n}{2}$:
\begin{equation}
\zeta(\Lambda\land t)\geq \frac{2}{n} \exp\Big((\Lambda\land t)a\Big),  
\label{eq:z}
\end{equation}
and directrl from the above we conclude that:
\begin{equation}
\mathbb{E}\Big(\exp\big((\Lambda\land t)a\big) \Big)\leq \frac{n}{2}\mathbb{E}\Big(\zeta(\Lambda\land t)\Big)\leq\frac{n}{2},  
\end{equation}
where in the last step we used the supermartingale property. Moreover, as $\exp((\Lambda\land t)a)$ converges monotonically to
$\exp(\Lambda a)$ as $t\to \infty$, we have also:
\begin{equation}
\mathbb{E}\Big( \exp(\Lambda a) \Big)\leq \frac{n}{2}.  
\end{equation}
Let us recall the Markov's inequality (see Theorem \ref{thm:me}):
\begin{equation}
  \mathbb{P}(|X|\geq c )\leq \frac{\mathbb{E}(|X|)}{c} 
\end{equation}
and choose $t_1=\frac{1}{a}(\ln(n)-\ln(\epsilon))$. Then we finally get:
\begin{equation}
\mathbb{P}(\Lambda>t_1)=\mathbb{P}\Big(\exp(\lambda a)>\frac{n}{\epsilon}\Big)\leq \frac{ \mathbb{E}({\exp(\lambda a)})}{\frac{n}{\epsilon}}\leq\frac{\epsilon}{2}.  
\end{equation}
\newline
Let us now consider $|I(t)|\geq\frac{n}{2}$. For this case we perform exactly the same procedure, but starting from Eq. (\ref{eq:gend}) instead of Eq. (\ref{eq:dfi}), which we started with in the previous case. We also redefine $\Lambda=\inf\{t:|I(t)|>I_s\}$. Following the same steps as above we only change constant $a$ in Eq. (\ref{eq:where}) into $b=\frac{1}{d^2}(\frac{2z}{I_s} \Phi(P)-\frac{2z}{n} \Phi(P)-r)$. Second thing that has to be changed is Eq. (\ref{eq:z}) where, instead of $\frac{n}{2}$ we can put $n$. Resulting time for this stage is:
\begin{equation}
\mathbb{P}(\Lambda>t_2)\leq \frac{ \mathbb{E}({\exp(\lambda b)})}{\frac{n}{\epsilon}}\leq\epsilon,\,\,  
 \textrm{where}\,\,\,t_2=\frac{1}{b}(\ln(n)-\ln(\epsilon)).
\end{equation}

\end{proof}

From this general theorem we conclude, that the closer we are with chosen parameters to the zero-stationary state (i.e. the smaller is the quasi-stationary density of infected nodes), the slower is the first phase of rapid increase:

\begin{remark}
Mixing time is linear with inverse of the distance $\eta$ from the epidemic threshold, i.e.:
$$ T(\eta,\epsilon)=O\Big(\frac{1}{\eta}(\log n + \log \epsilon^{-1})\Big). $$
\label{cor:1}
\end{remark}

\begin{proof}
Recall Eq. (\ref{eq:d}): we demand $D\big(t\big|I(t)\big)\geq0$ and transform this condition to:
\begin{equation}
  \frac{z}{nr}\geq\frac{|I(t)|}{\Big(\sum\limits_{k\in I(t),\,j\notin I(t)}p_{jk}+\sum\limits_{k\in I(t),\,j\notin I(t)}p_{kj}\Big)},
\end{equation}
which boils down to equality for quasi-stationary state. We denote right hand side of this equation by $p_c$ for the case of epidemic threshold. Now let us take values of parameters $z$, $n$ and $r$ that take us a bit higher than threshold:
\begin{equation}
  \frac{z}{nr}=p_c(1+\eta),
\label{eq:para}
\end{equation}
where $\eta\geq0$. Let us recall some parts of the proof of Theorem \ref{thm:1}. Actually, all we have to do is to rewrite condition for $D\big(t\big|I(t)\big)$ in parametrization given in Eq. (\ref{eq:para}) and notion of $p_c$:
\begin{equation}
D\big(t\big|I(t)\big)\geq r\Big(\frac{z}{nr} \frac{1}{p_c} |I(t)| -|I(t)|\Big)=r\eta |I(t)|.
\end{equation}
We put this result into Eq. (\ref{eq:est}) obtaining:
\begin{equation}
\mathbb{E}\Big( \frac{1}{|I(t+1)|}\Big| I(t), \tau>t\Big)\leq \frac{1}{|I(t)|} \exp(-\frac{r\eta}{d^2}),
\end{equation}
and then we proceed in the same way as in the proof of Theorem \ref{thm:1}. The result is
\begin{equation}
\mathbb{P}(\Lambda>t_c)=\leq \frac{ \mathbb{E}({\exp(\lambda \frac{r\eta}{d^2} )})}{\frac{n}{\epsilon}}\leq\epsilon,\,\, 
\textrm{where}\,\,\,t_c=\frac{d^2}{r\eta}(\ln(n)-\ln(\epsilon)).
\end{equation}
   
\end{proof}

\section{Simulation}
\label{sim}
Here we present simulations for various types of networks, i.e. complete graphs, $G(n,p)$ random graphs, small-world graphs and graphs with power law degree distribution (scale-free networks). Computer-simulational investigations focus on the topics described theoretically in the latter section, i.e. epidemic threshold, quasi-stationary state and mixing time.

\subsection{Epidemic threshold}
We check here the behavior of the process in the very beginning, i.e. exactly at the first time step. Four kinds of networks are being examined: complete graph, $G(n,p)$ random graph with $p=0.5$, small world graph with $k=6$ neighbors on the circle and rewiring probability $p=0.5$ and scale-free network with the exponent $\alpha=2.5$. We vary sizes of networks $n$ and for each type of the graph we choose different recovery probability $r$. Looking for the critical value of contamination probability $z_c$ we change parameter $z$ and check for which value the fraction of infected nodes starts to increase. This procedure is being repeated 100000 times. Results are presented in Fig. \ref{fig:1step}. Visibly, simulations follow the theoretical prediction of Theorem \ref{thm:trsh} prefectly for all four kinds of graphs being examined.

\begin{center}
\begin{figure}
  \includegraphics[height=6.5cm]{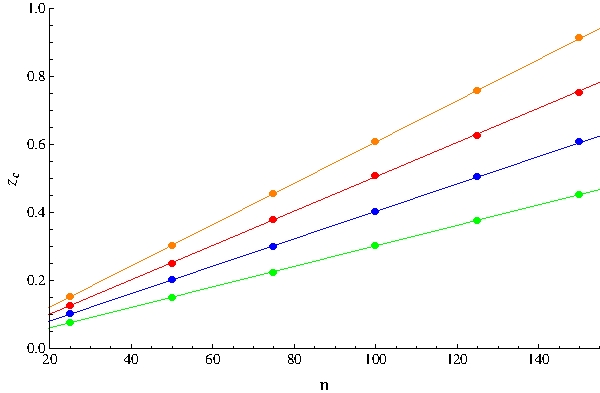}
  \caption{Epidemic threshold for four different type of graphs: dots stay for simulational results, lines present theoretical prediction, Theorem \ref{thm:trsh}. Starting from the bottom we have results for $G(n,p)$ random graph (green line, $r=0.006$), scale-free network (blue line, $r=0.008$), complete graph (red line, $r=0.01$) and small world graph (orange line, $r=0.012$).}
\label{fig:1step}
\end{figure}
\end{center}

\subsection{Quasi-stationary state}
 Results for quasi-stationary state are obtained by performing many runs (typically 1000), finishing each of them at a fixed, long time step (10 000 - 100 000), cutting the beginning phase of rapid increase and fitting a line to the points oscillating about the quasi-stationary value of infected nodes density. There are two types of results possible to obtain after a single run: epidemic either dies at a certain point (i.e. number of infected nodes, due to fluctuations, reaches zero and -- by construction of the model -- stays zero, usually it happens at the very beginning of the process) or the number of infected nodes increases rapidely in the first stage, and then oscillates about some fixed value (see Fig.\ref{fig:singlerun}). In order to compute average quasi-stationary value of infected nodes density we neglect all the runs where there exists such a time step, when the number of infected nodes equals zero (this is due to condition $\tau>t$). 

First we examine complete graphs: in Sec. \ref{cgsec} we provided the exact result for them, Eq. (\ref{eq:fg-stat}). In Fig. \ref{fig:fg-stat-n} we show how quasi-stationary infected nodes density $i_s$ depends on the network size $n$. Then, in Fig.\ref{fig:fg-stat-n}, we show dependence on the contamination probability $z$. Both figures show perfect agreement between simulation and theory.

As we have already seen the behavior of complete graphs and how they relate to the theory described above, let us compare $i_s$ for four different kinds of graphs. In Fig. \ref{fig:ogolne} we show the results for complete graph, $G(n,p)$ random graph with $p=0.1$, small world graph with $k=10$ neighbors on the circle and rewiring probability $p=0.5$ and scale-free network with the exponent $\alpha=2.5$. Sizes of the graphs are fixed, $n=100$. Noticeably, the results for three out of four kinds of graphs are almost the same, while scale-free network goes an entirely different way. Below we will focus on complete, $G(n,p)$ and small world graphs only. In Sec. \ref{largen} we concluded, that $G(n,p)$ graphs for large $n$ should resemble like complete graphs. It is instructive to see that in the limit of large $n$ not only on $G(n,p)$, but also on small world graphs the epidemy behaves the same as on complete graphs, see Fig. \ref{fig:srednieodn}.

\begin{center}
\begin{figure}[!h]
  \includegraphics[height=5.5cm]{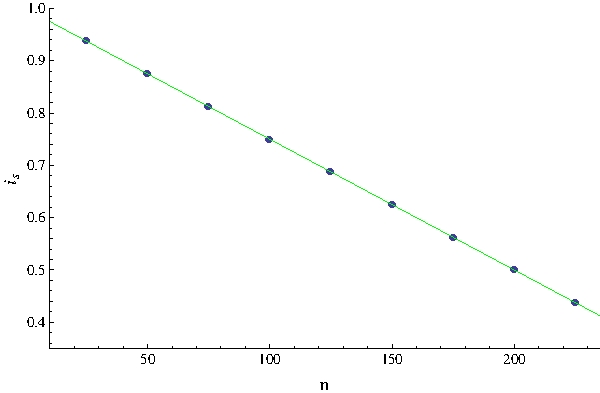}
  \caption{Plot of quasi-stationary value of infected nodes denstiy $i_s$ for complete graphs versus network size $n$: simulation (blue dots) and theoretical result (\ref{eq:fg-stat}) (green line). We fix here $z$=1, $r$=0.005.}
\label{fig:fg-stat-n}
\end{figure}
\end{center}

\begin{center}
\begin{figure}[!h]
  \includegraphics[height=5.5cm]{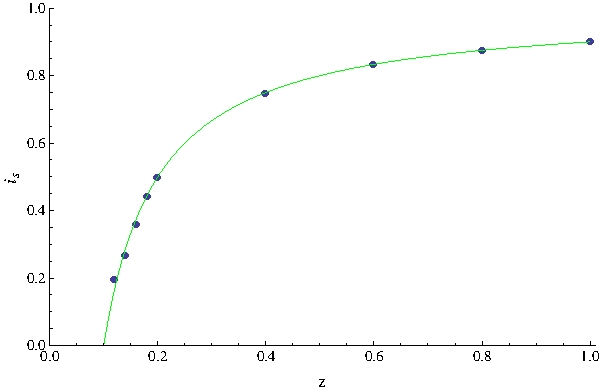}
  \caption{Plot of quasi-stationary value of infected nodes denstiy $i_s$ for complete graphs versus contamination probability $z$: simulation (blue dots) and theoretical result (\ref{eq:fg-stat}) (green line). We fix here $n$=100, $r$=0.002.}
  \label{fig:fg-stat-z}
\end{figure}
\end{center}

\begin{center}
\begin{figure}[!h]
  \includegraphics[height=5.5cm]{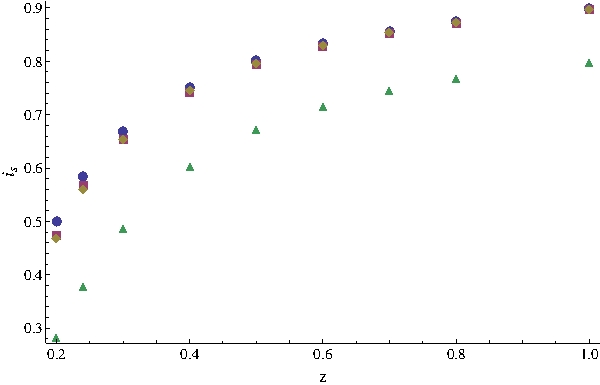}
  \caption{Plot of quasi-stationary value of infected nodes denstiy $i_s$ for complete graph (blue circles), $G(n,p)$ random graph with $p=0.1$ (red squares), small world graph with $k=10$ neighbors on the circle and rewiring probability $p=0.5$ (yellow diamonds) and scale-free network with the exponent $\alpha=2.5$ (green triangles) versus contamination probability $z$ and fixed network size $n=100$.}
\label{fig:ogolne}
\end{figure}
\end{center}

\begin{center}
\begin{figure}[!h]
  \includegraphics[height=6.5cm]{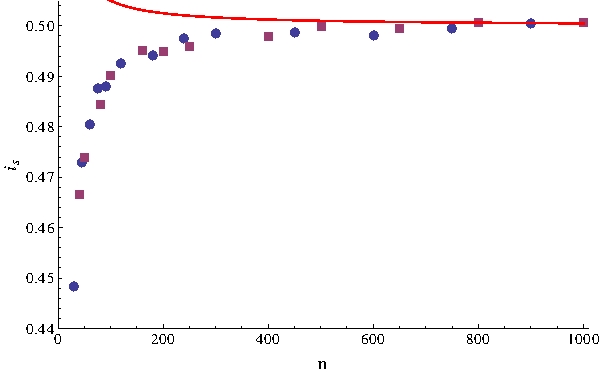}
  \caption{Plot of quasi-stationary value of infected nodes denstiy $i_s$ for $G(n,p)$ random graph with $p=0.2$ (blue circles) and small world graph with rewiring probability $p=0.5$ (red squares) versus network size $n$. Number of neighbors on the circle $k=2n/10$ is chosen such that the edges density $\frac{k}{2n}$ stays fixed. Red line shows theoretical prediction for complete graphs (\ref{eq:fg-stat}). We fix here $z=1$ and $n\times r=1$.}
\label{fig:srednieodn}
\end{figure}
\end{center}
\subsection{Mixing time}
In this section we examine mixing times of the process, i.e. we check how long does it take to reach quasi-stationary state. Fig. \ref{fig:czasy_n} depicts how does average mixing time depend on $ln(n)$, where $n$ is network size, as usual. This is done for complete graph, $G(n,p)$ random graph with $p=0.2$ and small world graph with rewiring probability $p=0.5$. For the same graphs we check average mixing time dependence on inverse of distance from epidemic threshold $\eta$ (see Corrolary \ref{cor:1}). It is shown in Fig. \ref{fig:czasy_e}. These result show actually much more than Theorem and Corollary from Sec. \ref{ms}. We examine here average mixing time and show, that they are linear with $ln(n)$ and $1/\eta$, as theory in Sec. \ref{ms} suggest by bounds of probability of mixing time proportional to $ln(n)$ and $1/\eta$.
\begin{center}
\begin{figure}[!h]
  \includegraphics[height=6.5cm]{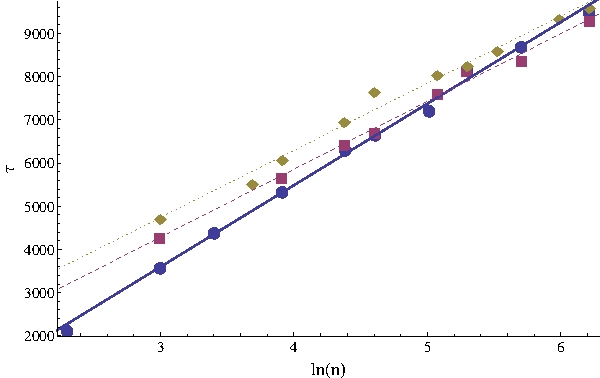}
  \caption{Average mixing time $\mathcal{T}$ for complete graph (blue dots), $G(n,p)$ random graph with $p=0.2$ (red squares), small world graph with rewiring probability $p=0.5$ (yellow diamonds) versus logarithm of network size $ln(n)$. Number of neighbors on the circle $k=2n/10$ is chosen such that the edges density $\frac{k}{2n}$ stays fixed. We fix here $r=0.001$ and $n/z=1000$ in order to have quasi-stationary state not changed. Lines are plotted to guide the eye.}
\label{fig:czasy_n}
\end{figure}
\end{center}

\begin{center}
\begin{figure}[!h]
  \includegraphics[height=6.5cm]{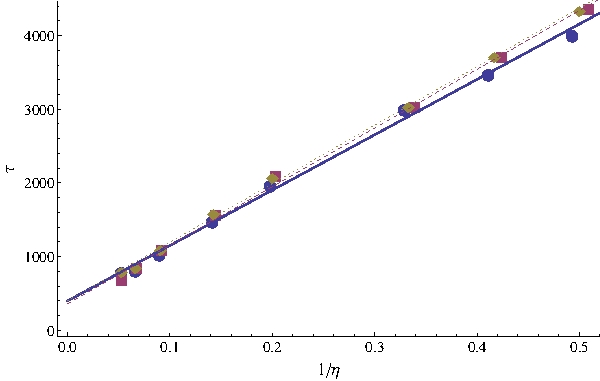}
  \caption{Average mixing time $\mathcal{T}$ for complete graph (blue dots), $G(n,p)$ random graph with $p=0.2$ (red squares), small world graph with rewiring probability $p=0.5$ and $k=20$ neighbors on the circle (yellow diamonds) versus inverse of distance from epidemic threshold $\eta$.  We fix here $r=0.001$ and $n=100$. Lines are plotted to guide the eye.}
\label{fig:czasy_e}
\end{figure}
\end{center}

\section{Conclusions}
\label{conclu}
We have proposed model for epidemic spreading with at most one infection per times step. Starting from the general formula for the change of the number of infected nodes (\ref{eq:d}) we provided condition for epidemic threshold for any kind of graph. Simulational results for epidemic threshold follow the theoretical predictions perfectly. Furthermore, quasi-stationary density of infected nodes for complete and uncorrelated homogenous graphs has been derived and bounds for this density, using the notion of graph conductance, have been obtained. Complete graph simulations show agreement with the theory. Epidemy on $G(n,p)$ random graphs, according to no-correlation in large $n$ limit \cite{gnp-nocor}, as well as on small-world graphs, in the large $n$ limit, behave like epidemy on complete graphs.
\newline\indent
We have proven theorem and corollary that bound the probability of mixing time by values proportional to $ln(n)$ and $1/\eta$, where $n$ and $\eta$ are size of the network and distance form  epidemic threshold respectively. Simulations on complete, $G(n,p)$ and small world graphs show even more, namely that the average mixing time is linear with $ln(n)$ and $1/\eta$.
%see: doiuf, chapter 8, relation to contact process

%\chapter*{Summary}
%\label{summary}
%\addcontentsline{toc}{chapter}{Summary}

\end{document}